\documentclass[twocolumn]{autart}    
\usepackage{blindtext}
\usepackage{graphicx}
\usepackage[utf8]{inputenc}
\usepackage{graphics} 
\usepackage{epsfig} 
\usepackage{amssymb}  
\usepackage{cite}
\usepackage{mathtools, cuted}
\usepackage[cmintegrals]{newtxmath}
\usepackage{epstopdf}
\usepackage{graphics} 
\usepackage{epsfig}

\usepackage{helvet}
\usepackage{courier}
\usepackage{subfigure}
\usepackage{type1cm}
\usepackage[utf8]{inputenc}
\usepackage[english]{babel}
\usepackage{makeidx}         
\usepackage{algorithmic}
\DeclareMathOperator*{\argmin}{\arg\!\min}

\usepackage{nomencl}
\usepackage[usenames, dvipsnames]{color}
\theoremstyle{definition}

\theoremstyle{theorem}
\newtheorem{theorem}{Theorem}

\theoremstyle{lemma}
\newtheorem{lemma}{Lemma}

\theoremstyle{remark}
\newtheorem{remark}{Remark}

 \theoremstyle{proposition}
\newtheorem{proposition}{Proposition}
\begin{document}

\begin{frontmatter}
\journal{}   
\title{Resilient Continuum Deformation Coordination
\thanksref{footnoteinfo}
} 

\thanks[footnoteinfo]{Authors are with the Department
of Aerospace Engineering, University of Michigan, Ann Arbor,
MI, 48109 USA e-mail: hosseinr@umich.edu.}

\author[Paestum]{Hossein Rastgoftar}
 \author[Rome]{Ella Atkins}

 \address[Paestum]{Department
of Aerospace Engineering, University of Michigan, Ann Arbor,
MI, 48109 USA}  
 \address[Rome]{Department
of Aerospace Engineering, University of Michigan, Ann Arbor,
 MI, 48109 USA
 }             


\begin{abstract}                          
This paper applies the principles of continuum mechanics to safely and resiliently coordinate {\color{black}a multi-agent team}. A hybrid automation with two operation modes, Homogeneous Deformation Mode (HDM) and Containment Exclusion Mode (CEM), are developed to robustly manage group coordination in the presence of unpredicted agent failures. HDM becomes active when all agents are healthy, where the group coordination is defined by homogeneous transformation coordination functions. By classifying agents as leaders and followers, a desired $n$-D homogeneous transformation is uniquely related to the desired trajectories of $n+1$ leaders and acquired by the remaining followers in real-time through local communication. The paper offers a novel approach for leader selection as well as naturally establishing  and reestablishing inter-agent communication whenever the agent team enters the HDM. CEM is activated when at least one agent fails to admit group coordination. This paper applies unique features of decentralized homogeneous transformation coordination to  quickly detect each arising anomalous situation and excludes failed agent(s) from group coordination of healthy agents. In CEM, agent coordination is treated as an ideal fluid flow where the desired agents' paths are defined along stream lines inspired by fluid flow field theory to circumvent exclusion spaces surrounding failed agent(s). 
\end{abstract}

\end{frontmatter}
\section{Introduction}
Control of multi-agent systems has been widely investigated over the past two decades. Formation and cooperative control can reduce cost and improve the robustness and capability of reconfiguration in a cooperative mission. Therefore, researchers have been motivated to explore diverse applications for the multi-agent coordination such as formation control \cite{xiao2009finite}, traffic congestion control \cite{wiering2000multi}, distributed sensing, \cite{li2009distributed}, cooperative surveillance \cite{wu2010multi}, and cooperative payload transport \cite{michael2011cooperative}. 
\subsection{Related Work}
Centralized and decentralized cooperative control approaches have been previously proposed for multi-agent coordination. The virtual structure \cite{ren2004decentralized}\cite{ren2002virtual} model treats agents as particles of a rigid body. Assuming the virtual body has an arbitrary translation and rigid body rotation in a $3$-D motion space, the desired trajectory of every agent is determined in a centralized fashion. Consensus \cite{hou2017consensus}\cite{ lin2009consensus}\cite{kim2014leaderless}\cite{defoort2015leader}\cite{papachristodoulou2010effects}\cite{cepeda2011exhaustive}\cite{zuo2016distributed}\cite{lin2016finite}\cite{su2018semi}\cite{yu2015finite} and containment control are the most common decentralized coordination approaches.  
Multi-agent coordination using first-order consensus \cite{liu2011stationary} and second-order consensus \cite{hou2017consensus}\cite{lin2009consensus} has been extensively investigated by researchers in the past.
Leader-based and leaderless consensus have been studied in Refs. \cite{kim2014leaderless} and \cite{defoort2015leader}. Stability of the retarded consensus method was studies in Refs. \cite{papachristodoulou2010effects}\cite{cepeda2011exhaustive}. Finite-time multi-agent consensus of continuous time systems is studied in Refs. \cite{zuo2016distributed}\cite{lin2016finite}. Refs. \cite{su2018semi}\cite{yu2015finite} evaluate consensus  under a switching communication topology in the presence of disturbances. 

More recently, researchers have investigated the resilient consensus problem and provided guarantee conditions for reaching  consensus in the presence of malicious agents \cite{leblanc2012resilient, shang2018resilient, dibaji2017resilient, leblanc2013resilient}. Weighted Mean Subsequence Reduced (W-MSR) is commonly used to detect an adversary and remove malicious agent(s) from the communication network of normal agents \cite{dibaji2017resilient, leblanc2013resilient}. $r$-robustness and $(r,s)$-robustness conditions are used to prove network resilience under consensus. Particularly, $(f+1,f+1)$-robustness is considered as the necessary and sufficient condition for resilience of the consensus protocol in the presence of $f$ malicious agents \cite{dibaji2017resilient}.

Containment control is a decentralized leader-follower approach in which multi-agent coordination is guided by a finite number of leaders and acquired by followers through local communication.  Necessary and sufficient conditions for the stability of continuum deformation coordination have been provided in Refs. \cite{liu2012necessary}. Ref. \cite{ji2008containment} studies the convergence of containment control and demonstrates that followers ultimately converge to the convex hull defined by leaders. Containment control under fixed and switching communication protocols are studied in Refs. \cite{cheng2017event} and \cite{li2015containment}, respectively. Refs. \cite{wang2013distributed, zhao2013distributed} study  finite-time containment stability and convergence. Containment control stability in the presence of communication delay is studied in Refs. \cite{shen2016containment, xiong2018containment}.

Continuum deformation for large-scale coordination of multi-agent systems is developed in \cite{rastgoftar2016continuum}. Similar to containment control, continuum deformation is a leader-follower approach in which a group coordination is guided by a finite number of leaders and acquired by followers through local communication \cite{rastgoftar2014evolution}. Because continuum deformation defines a non-singular mapping between reference and current agent configurations at any time $t$, follower communication weights are consistent with leader agents' reference positions in the continuum deformation coordination. The continuum deformation method advances containment control by formal characterization of safety in a large-scale coordination. Assuming continuum deformation is given by a homogeneous transformation, inter-agent collision and agent follower containment are guaranteed in a continuum deformation coordination by assigning a lower limit on the eigenvalues of the Jacobian matrix of the homogeneous transformation. Therefore, a large number of agents participating in a continuum deformation coordination can safely and aggressively deform to pass through narrow passages in a cluttered environment. 


\subsection{Contributions and Outline}
This paper proposes a physics-inspired approach to the resilient multi-agent coordination problem. In particular, multi-agent coordination is modeled by a hybrid automation with two physics-based coordination modes: (i) Homogeneous Deformation Mode (HDM) and (ii) Containment Exclusion Mode (CEM). 

HDM is active when all agents are healthy and can admit the desired coordination defined by a homogeneous deformation. In the HDM, agents are treated as particles of an $n$-D deformable body and the desired coordination, defined based on the trajectories of $n+1$ leaders forming an $n$-D simplex at any time $t$, is acquired by the followers through local communication.\footnote{This paper considers agents as particles of $2$-D and $3$-D deformable bodies where $2$-D and $3$-D homogeneous deformation defines the desired coordination at the HDM. Therefore, $n$ is either $2$ or $3$.}.  

CEM is activated once an adversarial situation is detected due to unpredicted vehicle or agent failure. The paper offers a novel approach for rapid detection of each anomalous or failed agent and excludes it from group coordination with the healthy vehicles. In CEM the desired coordination is treated as an irrotational fluid flow and adversarial agents are excluded from the safe planning space by combining ideal fluid flow patterns in a computationally-efficient manner. 


%
%
%
%
Compared to the existing literature and the authors' previous work, this paper offers the following contributions:
\begin{enumerate}
\item{The paper offers a novel distributed approach for detection of anomalous situations in which unexpected vehicle failure(s) disrupt collective vehicle motion.}
\item{This paper advances the existing continuum deformation coordination theory by relaxing the follower containment constraint and offering a tetrahedralization approach to assign followers' communication weights in an unsupervised fashion.}
\item{The paper proposes a model-free guarantee condition for convergence and inter-agent collision avoidance in a large-scale homogeneous transformation.}
\item{Compared to existing resilient coordination work \cite{leblanc2012resilient}\cite{shang2018resilient}\cite{dibaji2017resilient}\cite{leblanc2013resilient}, this paper offers a computationally-efficient safety recovery method. At CEM, every agent assigns its own desired trajectory without communication with other agents only by knowing the geometry of the unsafe domains enclosing the anomalous agents, as well as its own reference position when the CEM is activated.}
\item{This paper proposes a tetrahedralization method to (i) naturally establish/reestablish inter-agent communication links and weights, (ii) classify agents as boundary and follower agents, (iii) determine leaders in an unsupervised fashion.}
\item{{The authors believe this is the first paper describing safe exclusion of a failed agent in a cooperative team with inspiration from fluid flow models.}}
\end{enumerate}

This paper is organized as follows.  Preliminaries in Section \ref{Preliminaries} are followed by a resilient continuum deformation formulation in Section \ref{Problem Formulation and Statement}. Physics-based models for the HDM and CEM are described in Section \ref{Physics-based Modeling of the HDM and CEM}. Operation of the resilient continuum deformation coordination is modeled by a hybrid automation in Section \ref{Continuum Deformation Anomaly Management}. Simulation results presented in Section \ref{results} are followed by concluding remarks in Section \ref{conclusion}.
\section{Preliminaries}
\label{Preliminaries}
\subsection{Position Notations}
The following position notations are used throughout this paper:
{\underline{Reference position} of vehicle $i$ is denoted by $\mathbf{r}_{i,0}=\left[x_{i,0}~y_{i,0}~z_{i,0}\right]^T$. \textit{In this paper, a reference configuration is defined based on agents' current positions once they enter the HDM.}}
{\underline{Actual position} of vehicle $i$ is denoted $\mathbf{r}_i(t)=[x_i(t)~y_i(t)~z_i(t)]^T$ at time $t$. }
{\underline{Local desired position} of vehicle $i$ is denoted by $\mathbf{r}_{i,d}(t)=[x_{i,d}(t)~y_{i,d}(t)~z_{i,d}(t)]^T$ at time $t$. $\mathbf{r}_{i,d}$ is updated through local communication and defined based on actual positions of the in-neighbor vehicles of agent $i$.}
{\underline{Global desired position} of vehicle $i$ is denoted by $\mathbf{r}_{i,c}(t)=[x_{i,c}(t)~y_{i,c}(t)~z_{i,c}(t)]^T$ at time $t$. The transient error is defined as the difference between actual position $\mathbf{r}_i(t)$ and global desired position  $\mathbf{r}_{i,c}(t)$ for vehicle $i$ at any time $t$.}

\subsection{{Motion Space Tetrahedralization and $\mathbf{\Lambda}$ Operator}} 
 \label{Motion Space Tetrahedrlization and Operator}
Assume $\mathbf{c}\in\mathbb{R}^{3\times 1}$; $\mathbf{p}_1$, $\mathbf{p}_2$, $\cdots$, $\mathbf{p}_{n+1}\in \mathbb{R}^{3\times 1}$ are $n+1$ arbitrary position vectors in a $3$-D motion space. Then, rank operator $\varkappa_n$ is defined as follows:
\begin{equation}
n=2,3,\qquad    \varkappa_n\left(\mathbf{p}_1,\cdots,\mathbf{p}_n\right)=\mathrm{rank}\left(
    \begin{bmatrix}
        \mathbf{p}_2-\mathbf{p}_1&\cdots&\mathbf{p}_{n+1}-\mathbf{p}_1
    \end{bmatrix}
    \right).
\end{equation}
If $\mathbf{p}_1$, $\cdots$, $\mathbf{p}_{n+1}$ are positioned at the vertices an $n$-D simplex, then $ \varkappa_n\left(\mathbf{p}_1,\cdots,\mathbf{p}_n\right)=n$. $\mathbf{p}_1$, $\mathbf{p}_2$, $\mathbf{p}_3$, and $\mathbf{p}_4$ form a tetrahedron for $n=3$, if $\varkappa_3\left(\mathbf{p}_1,\mathbf{p}_2,\mathbf{p}_3,\mathbf{p}_4\right)=3$. 
$\mathbf{p}_1$, $\mathbf{p}_2$, and $\mathbf{p}_3$ form a triangle for $n=2$, if $\varkappa_2\left(\mathbf{p}_1,\mathbf{p}_2,\mathbf{p}_3\right)=2$. 
\\
\textbf{Operator $\mathbf{\Lambda}$ ($n=2,3$):} Assume $\mathbf{p}_1$, $\mathbf{p}_2$, $\mathbf{p}_3$, and $\mathbf{p}_4^n$ are known points in a $3$-D motion space, where
\begin{equation}
    \varkappa_3\left(\mathbf{p}_1,\mathbf{p}_2,\mathbf{p}_3,\mathbf{p}_4^n\right)=3.
\end{equation}

Then, operator $\mathbf{\Lambda}$ can be defined as follows:
\begin{equation}
n=2,3,\qquad    \mathbf{\Lambda}\left(\mathbf{p}_1,\mathbf{p}_2,\mathbf{p}_3,\mathbf{p}_4^n,\mathbf{c}^n\right)=
    \begin{bmatrix}
        \mathbf{p}_1&\mathbf{p}_2&\mathbf{p}_{3}&\mathbf{p_4}^n\\
        1&1&1&1
    \end{bmatrix}
    ^{-1}\begin{bmatrix}
        \mathbf{c}^n\\
        1
    \end{bmatrix}
    ,
\end{equation}
where $\mathbf{c}^n$ is the position of a point in a $3$-D motion space.
If $\varkappa_3\left(\mathbf{p}_1,\mathbf{p}_2,\mathbf{p}_3,\mathbf{p}_4^n\right)=3$, $\mathbf{\Lambda}\left(\mathbf{p}_1,\cdots,\mathbf{p}_n,\mathbf{c}^n\right)$ exists and has the following properties \cite{rastgoftar2019formal}:
\begin{enumerate}
    \item{The sum of the entries of $\mathbf{\Lambda}$ is $1$ for any configuration of vectors $\mathbf{p}_1,\cdots,\mathbf{p}_{4}^n$ and $\mathbf{c}^n$ for $n=2,3$.}
    \item{If $\mathbf{\Lambda}>\mathbf{0}$, $\mathbf{c}^n$ is inside the tetrahedron formed by $\mathbf{p}_1$, $\cdots$, $\mathbf{p}_{4}^n$. Otherwise, it is outside the tetrahedron.}
\end{enumerate}
\textbf{Motion Space Tetrahedralization}: A $3$-D motion space can be divided into two subspaces that are inside and outside of the tetrahedron defined by vectors $\mathbf{p}_1$, $\mathbf{p}_2$, $\mathbf{p}_3$, and $\mathbf{p}_4^n$, if $\varkappa_3\left(\mathbf{p}_1,\mathbf{p}_2,\mathbf{p}_3,\mathbf{p}_4^n\right)=3$.
\begin{figure}[ht]
\centering
\includegraphics[width=3.3 in]{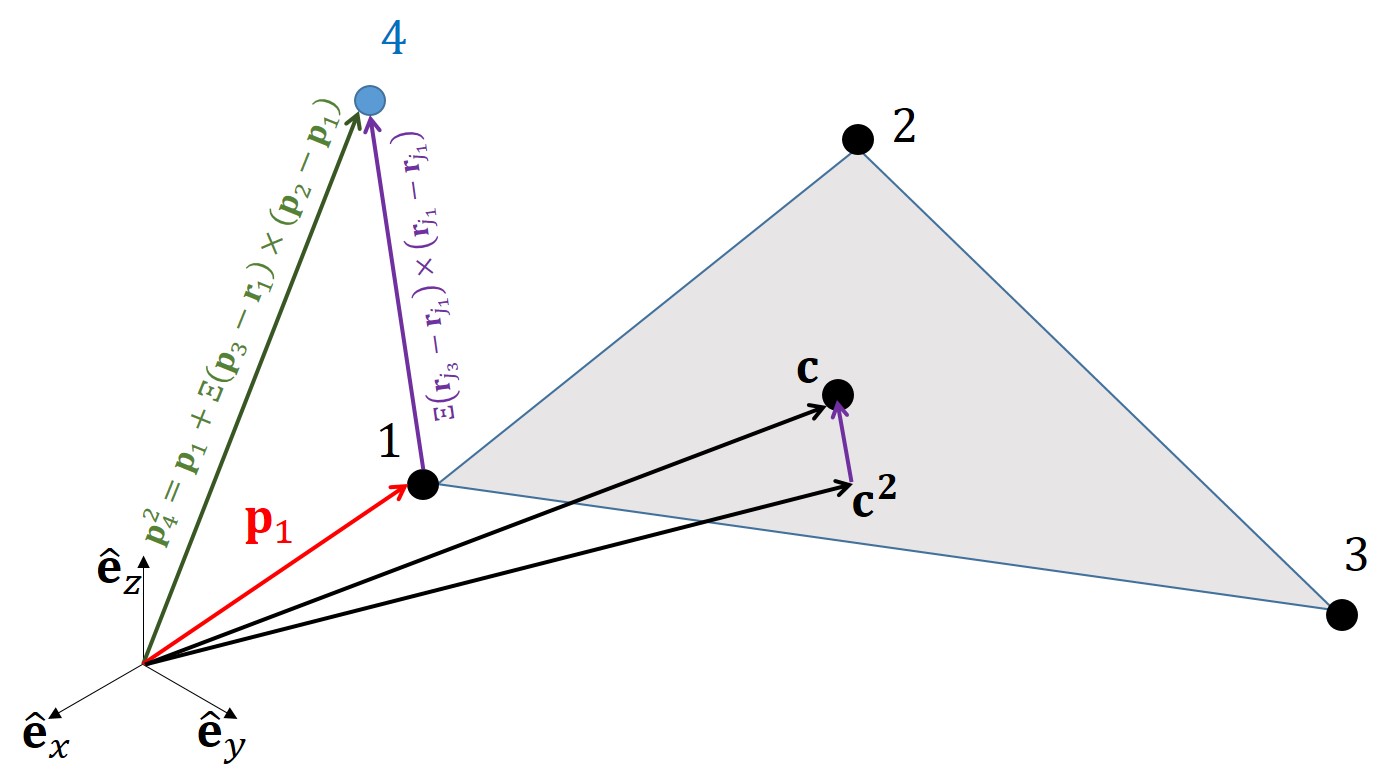}
\caption{Graphical representation of the virtual agent $\mathbf{p}_{4}^2$.  }
\label{Detecton2D}
\end{figure}

\underline{For $n=3$}, $\mathbf{p}_1$, $\mathbf{p}_2$, $\mathbf{p}_3$, $\mathbf{p}_4$, and $\mathbf{c}$ represent real points (agents) in a $3$-D motion space, if  $\varkappa_3\left(\mathbf{p}_1,\mathbf{p}_2,\mathbf{p}_3,\mathbf{p}_4\right)=3$. Therefore, $\mathbf{p}_4^3=\mathbf{p}_4$ and $\mathbf{c}^3=\mathbf{c}$.

\underline{For $n=2$}, $\mathbf{p}_1$, $\mathbf{p}_2$, and $\mathbf{p}_3$ are the real points forming a triangle in a $3$-D motion space. Given $\mathbf{p}_1$, $\mathbf{p}_2$, and $\mathbf{p}_3$, \textbf{virtual point} 
\begin{equation}
\label{virtualpoint}
    \mathbf{p}_{4}^2=\mathbf{p}_{1}+\Xi\left(\mathbf{p}_{3}-\mathbf{p}_{1}\right)\times \left(\mathbf{p}_{2}-\mathbf{p}_{1}\right),
\end{equation}
where $\Xi\neq 0$ is constant. Note that $\mathbf{p}_{4}^2$, defined by Eq. \eqref{virtualpoint}, is perpendicular to the triangular plane made by agents $\mathbf{p}_1$, $\mathbf{p}_2$, and $\mathbf{p}_3$ (see Fig. \ref{Detecton2D}). Consequently, virtual agent $\mathbf{p}_{4}^2$ and in-neighbor agent $\mathbf{p}_1$, $\mathbf{p}_2$, and $\mathbf{p}_3$ form a tetrahedron. The projection of $\mathbf{c}$ on the triangular plane made by $\mathbf{p}_1$, $\mathbf{p}_2$, and $\mathbf{p}_3$ is denoted by $\tilde{\mathbf{c}}^2$ and expressed as follows:
\begin{equation}
\label{projection}
\begin{split}
    {\mathbf{c}}^2={\mathbf{c}}-\left({\mathbf{c}}\cdot\mathbf{n}_{1-4}\left(\mathbf{p}_1,\mathbf{p}_2,\mathbf{p}_3\right)\right)\mathbf{n}_{1-4}\left(\mathbf{p}_1,\mathbf{p}_2,\mathbf{p}_3\right),
\end{split}
\end{equation}
where unit vector
\begin{equation}
\label{n144}
    \mathbf{n}_{1-4}\left(\mathbf{p}_1,\mathbf{p}_2,\mathbf{p}_3\right)=\dfrac{\left(\mathbf{p}_{3}-\mathbf{p}_{1}\right)\times \left(\mathbf{p}_{2}-\mathbf{p}_{1}\right)}{\|\left(\mathbf{p}_{3}-\mathbf{p}_{1}\right)\times \left(\mathbf{p}_{2}-\mathbf{p}_{1}\right)\|}
\end{equation}
is normal to the triangular plane made by $\mathbf{p}_1$, $\mathbf{p}_2$, and $\mathbf{p}_3$.

\begin{proposition}\label{prop11111111111}
Let $\mathbf{\Lambda}$ be expressed in component-wise form:
\[
\mathbf{\Lambda}=\begin{bmatrix}
      \lambda_1&\lambda_2&\lambda_3&\lambda_4
\end{bmatrix}
^T
.
\]
If $n=2$, $\lambda_4\left(\mathbf{p}_1,\mathbf{p}_2,\mathbf{p}_3,\mathbf{p}_4^2,\mathbf{c}^2\right)=0$ for any arbitrary position $\mathbf{c}^2$. 
\end{proposition}

\begin{proof}
Given $\mathbf{p}_1$, $\mathbf{p}_2$, $\mathbf{p}_3$, and $\mathbf{p}_4^n$, $\lambda_4\left(\mathbf{p}_1,\mathbf{p}_2,\mathbf{p}_3,\mathbf{p}_4^n,\mathbf{c}^n\right)$ is obtained as follows:
\begin{equation}
\label{lambda4}
    \lambda_4\left(\mathbf{p}_1,\mathbf{p}_2,\mathbf{p}_3,\mathbf{p}_4^n,\mathbf{c}^n\right)=\dfrac{\|\mathbf{c}^n-\mathbf{c}^2\|}{\|\mathbf{p}_4^n-\mathbf{p}_1\|}.
\end{equation}
For $n=2$, the denominator of Eq. \eqref{lambda4} is $0$, thus $\lambda_4\left(\mathbf{p}_1,\mathbf{p}_2,\mathbf{p}_3,\mathbf{p}_4^n,\mathbf{c}^n\right)=0$ for any arbitrary position of point $\mathbf{c}$ in the motion space.
\end{proof}
Operator $\mathbf{\Lambda}$ will be used to (i) determine boundary and interior agents, (ii) specify followers' in-neighbor agents, (iii) assign followers' communication weights in a $2$-D and $3$-D homogeneous deformation coordination, and (iv) detect anomalies in a group coordination.

\subsection{Homogeneous Deformation}
\label{Homogeneous Deformation}
A homogeneous deformation is an affine transformation\footnote{The affine transformation \eqref{HT} is called \textit{Homogeneous Deformation} in continuum mechanics \cite{lai2009introduction}.} given by
\begin{equation}
\label{HT}
    \mathbf{r}_{i,c}(t)=\mathbf{Q}(t)\mathbf{r}_{i,0}+\mathbf{d}(t),
\end{equation}
where $\mathbf{Q}(t)$ is the Jacobian matrix, $\mathbf{d}(t)$ is the rigid-body displacement vector, $\mathbf{r}_{i,0}$ is the reference position of agent $i\in \mathcal{V}_H(t)$, and $\mathcal{V}_H(t)$ defines index numbers of healthy agents at time $t$.

\underline{\textit{$\alpha$ parameters}:} Let $\mathcal{V}_H(t)$ be expressed as 
\begin{equation}
\label{HLF}
    \mathcal{V}_H=\mathcal{V}_L\bigcup \mathcal{V}_F.
\end{equation}
where $\mathcal{V}_L=\{i_1,\cdots,i_{n+1}\}$ and $\mathcal{V}_F(t)=\{i_{n+2},\cdots,i_{N_F}\}$ are disjoint sets defining leaders and followers at time $t$. Let $\mathbf{r}_{i_1,0}$, $\cdots$, $\mathbf{r}_{i_{n+1},0}$ denote the reference positions of the leaders and $\mathbf{r}_{i_j,0}$ denotes the reference position of follower $i_j\in \mathcal{V}$, where reference positions are all assigned at the time agents first enter HDM. Then, we can define $\alpha$ parameters $\alpha_{i_j,i_1}$ through $\alpha_{i_j,i_4}$ as follows:
\begin{equation}
    \begin{bmatrix}
          \alpha_{i_j,i_1}&\cdots&\alpha_{i_j,i_4}
    \end{bmatrix}
    ^T
    =\mathbf{\Lambda}\left(\mathbf{r}_{i_1,0},\mathbf{r}_{i_2,0},\mathbf{r}_{i_3,0},\mathbf{r}_{i_4,0}^n,\mathbf{r}_{i_j,0}^n\right),
\end{equation}
where
\begin{subequations}
 \label{9aaa}
 \begin{equation}
     \mathbf{r}_{i_4,0}^n=
     \begin{cases}
     \mathbf{r}_{i_4,0}&n=3\\
     \mathbf{r}_{i_1,0}+\Xi\left(\mathbf{r}_{i_3,0}-\mathbf{r}_{i_1,0}\right)\times \left(\mathbf{r}_{i_2,0}-\mathbf{p}_{i_1,0}\right)&n=2
     \end{cases}
     ,
 \end{equation}
 \begin{equation}
     \mathbf{r}_{i_j,0}^n=
     \begin{cases}
     \mathbf{r}_{i_j,0}&n=3\\
      {\mathbf{r}}_{i_j,0}-\left({\mathbf{r}}_{i_j,0}\cdot\mathbf{n}_{1-4}\right)\mathbf{n}_{1-4}&n=2
     \end{cases}
     ,
 \end{equation}
\end{subequations}
and $\mathbf{n}_{1-4}=\mathbf{n}_{1-4}\left(\mathbf{r}_{i_1,0},\mathbf{r}_{i_2,0},\mathbf{r}_{i_3,0}\right)$ was previously defined in \eqref{n144}. Note that $\alpha_{i_1,i_4}=0$, if $n=2$ (See Proposition \ref{prop11111111111}).

\underline{\textit{Global Desired Position:}} Because homogeneous deformation is a linear transformation, global desired position of vehicle $i_j$ can be either given by Eq. \eqref{HT} or expressed as a convex combination of the leaders' positions at any time $t$. \begin{equation}
\label{GloBDesPos}
    j\in \mathcal{V}_F,\qquad \mathbf{r}_{j,c}=\sum_{k=1}^{n+1}\alpha_{j,i_k}\mathbf{r}_{i_k,c}(t).
\end{equation}

\section{Problem Formulation and Statement}\label{Problem Formulation and Statement}
Consider a $3$-D motion space containing $M$ agents where every agent is uniquely identified by a number $i\in \mathcal{M}= \left\{1,\cdots,M\right\}$. It is assumed that $N(t)$ (out of $M$) agents are enclosed by a rigid-size \textit{containment domain} 
\begin{equation}
\label{containmentdoamin}
    \Omega_{\mathrm{con}}=\Omega_{\mathrm{con}}(\mathbf{r},\mathbf{r}_{\mathrm{con}}(t))\subset \mathbb{R}^3,
\end{equation}
at time $t$. Let  $\mathbf{r}_{\mathrm{con}}(t)\in \mathbb{R}^{3\times 1}$ be the nominal position of the containment domain given by
\begin{equation}
    \mathbf{r}_{\mathrm{con}}(t)=\sum_{i=1}^{N(t)}\beta_i\mathbf{r}_{i,c}(t).
\end{equation}
Note that $0\leq\beta_i<1$ is a scaling factor, $\sum_{i=1}^{N(t)}\beta_i=1$, and the size of  $\Omega_{\mathrm{con}}=\Omega_{\mathrm{con}}(\mathbf{r},\mathbf{r}_{\mathrm{con}}(t))\subset \mathbb{R}^3$ does not change over time. Identification numbers of the agents enclosed by $\Omega_{\mathrm{cont}}(t)$ are defined by set 
\begin{equation}
    \mathcal{V}(t)=\left\{i\in \mathcal{M}\big|\mathbf{r}_{i,c}(t)\in \Omega_{\mathrm{con}}(\mathbf{r},\mathbf{r}_{\mathrm{con}}(t))\right\}
\end{equation}
Agents enclosed by the containment region $\Omega_{\mathrm{con}}(\mathbf{r},\mathbf{r}_{\mathrm{con}}(t))$ can be classified as \textit{healthy} or \textit{anomalous} agents, where \textit{healthy} agents admit the group desired coordination while \textit{anomalous} agents do not.
\textit{Healthy} and \textit{anomalous} vehicles are defined by disjoint sets $\mathcal{V}_H$ and $\mathcal{V}_A$, respectively, where  $\mathcal{V}$ can be expressed as 
\begin{equation}
    \mathcal{V}=\mathcal{V}_H\bigcup \mathcal{V}_A,
\end{equation}
where $\mathcal{V}_H=\left\{i_1,\cdots,i_{N_F}\right\}$ and $\mathcal{V}_A=\{i_{N_F+1},\cdots,i_N\}$. 

This paper treats agents as particles of a deformable body where the desired trajectory of vehicle $j\in \mathcal{V}$ is given by 
\begin{equation}
\label{MAINNNNN}
    \dot{\mathbf{r}}_{j,c}\left(x_{j,c},y_{j,c},z_{j,c},t\right)=\mathbf{H}_{j,\gamma}\left(x_{j,c},y_{j,c},z_{j,c}\right)\dot{\mathbf{q}}_\gamma(t),
\end{equation}
where $\gamma$ is a discrete variable defined by finite set $\Gamma=\{\mathrm{CEM},\mathrm{HDM}\}$. Set $\Gamma$ specifies the collective motion operation mode. $\mathbf{r}_{j,c}(t)=[x_{j,c}(t)~y_{j,c}(t)~z_{j,c}(t)]^T$ is the global desired trajectory of vehicle $j$, $\mathbf{q}=\left[q_{1,\gamma}~\cdots~q_{m,\gamma}^\gamma(t)\right]^T$, $q_{1,\gamma}(t)$ through $q_{N,\gamma}(t)$ are the \textit{\textbf{generalized coordinates}} specifying the temporal behavior of the group coordination. Furthermore,
\[
j\in \mathcal{V},~\gamma\in \Gamma,\qquad     \mathbf{H}_{j,\gamma}=
    \begin{bmatrix}
    \mathbf{h}_{j,1,\gamma}&\cdots&\mathbf{h}_{j,m,\gamma}
    \end{bmatrix}
    \in \mathbb{R}^{3\times m}
\]
is the spatially-varying \textbf{\textit{shape matrix}}.  $\mathbf{h}_{j,1,\gamma}(x_{j,c},y_{j,c},z_{j,c})\in \mathbb{R}^{3\times1}$ through $\mathbf{h}_{j,m,\gamma}(x_{j,c},y_{j,c},z_{j,c})\in \mathbb{R}^{3\times 1}$ are the \textbf{\textit{shape functions}}. 

\textbf{HDM} ($\gamma=\mathrm{HDM}$) is active when $\mathcal{V}_A=\emptyset$. Therefore, $N_F(t)=N(t)$ agents defined by set $\mathcal{V}_H$  are all healthy. The HDM shape matrix $\mathbf{H}_{j,\mathrm{HDM}}$ is time-invariant (constant), where $j\in \mathcal{V}_H$. The HDM generalized coordinate vector $\mathbf{q}_{_\mathrm{HDM}}\in \mathbb{R}^{3\left(n+1\right)\times 1}$ specifies desired velocity components of all leaders guiding the group continuum deformation coordination. This paper develops a decentralized leader-follower approach using the tetrahdralization presented in Section \ref{Motion Space Tetrahedrlization and Operator}. By classifying agents as leaders and followers, $\mathbf{V}_H=\mathcal{V}_L\bigcup\mathcal{V}_F$ (See Eq. \eqref{HLF}).
Leaders, defined by $\mathcal{V}_L=\{i_1,\cdots,i_{n+1}\}$, move independently. Followers, defined by $\mathcal{V}_F=\{i_{n+2},\cdots,i_N\}$, acquire the desired coordination through local communication with leaders and other followers. The paper offers a tetrahedralization method to determine leaders and followers and define inter-agent communication among vehicles in an unsupervised fashion for an arbitrary reference configuration of agents. 

\textbf{CEM} ($\gamma=\mathrm{CEM}$) is activated once at least one anomalous agent is detected in which case $\mathcal{V}_A\neq \emptyset$. The CEM shape matrix $\mathbf{H}_{j,\mathrm{CEM}}$ ($j\in \mathcal{V}_H$) is spatially varying. In particular, the desired vehicle coordination of healthy vehicle $j\in \mathcal{V}_H$ is defined by an ideal fluid flow. For CEM, it is desired that (i) vehicle $j\in \mathcal{V}_H$ moves along the surface $z_{j,c}=z_j\left(x_{j,c},y_{j,c},t\right)$ and (ii) $x$ and $y$ components of the agent coordination are defined by an irrotational flow. Mathematically speaking, we define coordinate transformation
\begin{equation}
\label{CoordinateTransformation}
    \begin{cases}
    \phi_{j,c}=&\phi\left(x_{j,c},y_{j,c},t\right)\\
    \psi_{j,c}=&\psi\left(x_{j,c},y_{j,c},t\right)\\
    z_{j,c}=&z_j\left(x_{j,c},y_{j,c},t\right)
    \end{cases}
    ,
\end{equation}
 where $j\in \mathcal{V}_H$, $\phi(x_{j,c},y_{j,c},t)$ and $\psi(x_{j,c},y_{j,c},t)$ satisfy the Laplace equation: 
 \begin{subequations}
\begin{equation}
    \dfrac{\partial ^2 \phi\left(x_{j,c},y_{j,c},t\right)}{\partial x_{j,c}^2}+\dfrac{\partial ^2 \phi\left(x_{j,c},y_{j,c},t\right)}{\partial y_{j,c}^2}=0
\end{equation}
\begin{equation}
    \dfrac{\partial ^2 \psi\left(x_{j,c},y_{j,c},t\right)}{\partial x_{j,c}^2}+\dfrac{\partial ^2 \psi\left(x_{j,c},y_{j,c},t\right)}{\partial y_{j,c}^2}=0.
\end{equation}
\end{subequations}

For smooth "flow" every agent $j\in \mathcal{V}_H$ slides along the $j$-th streamline defined by 
$
j\in \mathcal{V}_H,\qquad \psi\left(x_{j,c},y_{j,c},t\right)=\psi_{j,0}=\mathrm{constant}
$
at any time $t$. 
This condition requires that the  desired trajectory of vehicle $j\in \mathcal{V}_H$  satisfy the following equation at any time $t$:
\begin{equation}
\label{StickingCondition}
    \dfrac{\partial \psi(x_{j,c},y_{j,c},t)}{\partial x_{j,c}}\dfrac{dx_{j,c}}{dt}+\dfrac{\partial \psi(x_{j,c},y_{j,c},t)}{\partial y_{j,c}}\dfrac{dy_{j,c}}{dt}=0.
\end{equation}
Notice that stream and potential functions satisfy the Cauchy-Riemann condition. Therefore, the level curves $\phi(x,y,t)=\mathrm{constant}$ and $\psi(x,y,t)=\mathrm{constant}$ are perpendicular at the intersection point. This paper defines $\phi(x,y,t)$ and $\psi(x,y,t)$ by combining ideal fluid flow patterns so that an obstacle-free motion space is excluded from adversaries. This combination can split the $x-y$ plane into a safe region defined by set $\mathcal{S}$ and unsafe region defined by set $\mathcal{U}$. A one-to-one mapping exists between $(x_j,y_j)$ and $(\phi(x_j,y_j),\psi(x_j,y_j))$ at every point in the safe set $\mathcal{S}$. Thus, the Jacobian matrix
\begin{equation}
\label{Jaaaacobian}
(x,y)\in \mathcal{S},~j\in \mathcal{V}_H\qquad 
    \mathbf{J}(x_j,y_j)=
    \begin{bmatrix}
    \dfrac{\partial \phi}{\partial x_j}&\dfrac{\partial \phi}{\partial y_j}\\
    \dfrac{\partial \psi}{\partial x_j}&\dfrac{\partial \psi}{\partial y_j}\\
    \end{bmatrix}
\end{equation}
is nonsigular. Potential and stream fields are generated by combining ``Uniform'' and ``Doublet'' flow patterns. As a result, a single failed vehicle can be separated by a cylinder from the safe region $\mathcal{S}$  in the motion space. 

This paper also offers a novel distributed anomaly detection approach by using the properties of leader-follower homogeneous transformation coordination. Particularly, the $\mathbf{\Lambda}$ operator is used to characterize agent deviation of agents from the desired coordination to quickly identify failed agent(s) that are not admitting the desired continuum deformation.

\section{Physics-based Modeling of HDM and CEM}
\label{Physics-based Modeling of the HDM and CEM}
HDM and CEM are mathematically modeled in this section. A decentralized leader follower method for HDM is developed in Section \ref{Homogeneous Deformation Mode} to acquire a desired continuum deformation in an unsupervised fashion. CEM coordination is modeled in Section \ref{Containment  Exclusion Mode (CEM)}. 
\subsection{Homogeneous Deformation Mode (HDM)}
\label{Homogeneous Deformation Mode}
In HDM, vehicles are healthy and cooperative. Therefore, $\left|\mathcal{V}_A\right|=0$ ($N_F=N$ and $\mathcal{V}_H=\mathcal{V}$).  Set $\mathcal{V}_{H}$ can be expressed as $\mathcal{V}_{H}=\mathcal{V}_L\bigcup\mathcal{V}_F$, where $\mathcal{V}_L=\{i_1,\cdots,i_{{n+1}}\}$ and $\mathcal{V}_F=\{i_{n+2},\cdots,i_{N}\}$ define leaders and followers, respectively. 

\subsubsection{{{Desired Homogeneous Deformation Definition}}} 
A desired homogeneous transformation can be defined by $m=3(n+1)$ generalized coordinates $q_{1,\mathrm{HDM}}$, $\cdots$, $q_{3(n+1),\mathrm{HDM}}$ using relation \eqref{MAINNNNN}, where
\begin{subequations}
\begin{equation}
\mathbf{q}_{\mathrm{HDM}}(t)=
    \begin{bmatrix}
    q_{1,\mathrm{HDM}}(t)\\
    \vdots\\
    q_{3(n+1),\mathrm{HDM}}(t)
    \end{bmatrix}
    =\mathrm{vec}\left(
    \begin{bmatrix}
    x_{i_1,c}&\cdots&x_{i_{n+1},c}\\
    y_{i_1,c}&\cdots&y_{i_{n+1},c}\\
    z_{i_1,c}&\cdots&z_{i_{m+1},c}\\
    \end{bmatrix}^T
    \right),
\end{equation}
\begin{equation}
\begin{split}
\mathbf{H}_{i_j,\mathrm{HDM}}=&
    \begin{bmatrix}
    \mathbf{h}_{j,1}&\cdots&\mathbf{h}_{j,3(n+1)}
    \end{bmatrix}
    =\mathbf{I}_3\otimes 
    \begin{bmatrix}
    \alpha_{j,i_1}&\cdots&\alpha_{j,i_{n+1}}
    \end{bmatrix},
\end{split}
\end{equation}
\end{subequations}
and $j\in \mathcal{V}$. Note that $\mathrm{vec}(\cdot)$ is the matrix vectorization operator.
In Section \ref{Homogeneous Deformation}, it was described how $\alpha$ parameters $\alpha_{j,i_1}$ through $\alpha_{j,i_{n+1}}$ are assigned based on the agents' reference positions.

\subsubsection{Unsupervised Acquisition of a Homogeneous Deformation Using Tetrahedralization}
\label{Natural Acquisition of a Homogeneous Deformation}
A desired homogeneous deformation, defined by $n+1$ leaders in an $n$-D homogeneous deformation, is acquired by followers through local communication. Communication among healthy agents is defined by \textbf{coordination graph} $\mathcal{G}_c\left(\mathcal{V},\mathcal{E}\right)$ with nodes $\mathcal{V}$ ($\mathcal{V}=\mathcal{V}_H$) and edges $\mathcal{E}\subset \mathcal{V}\times \mathcal{V}$. In-neighbor agents of agent $i\in \mathcal{V}$ are defined by
\[
i\in \mathcal{V},\qquad \mathcal{N}_i=\{j\big|(j,i)\in \mathcal{E}
\}.
\]
Assuming the reference formation of agents is known, $n+1$ boundary agents are selected as leaders. Furthermore, every follower $i\in \mathcal{V}_F$ communicates with $n+1$ in-neighbor agents where the in-neighbor agents are placed at the vertices of an $n$-simplex containing follower $i$. 
\begin{figure}[ht]
\centering
\includegraphics[width=3.3 in]{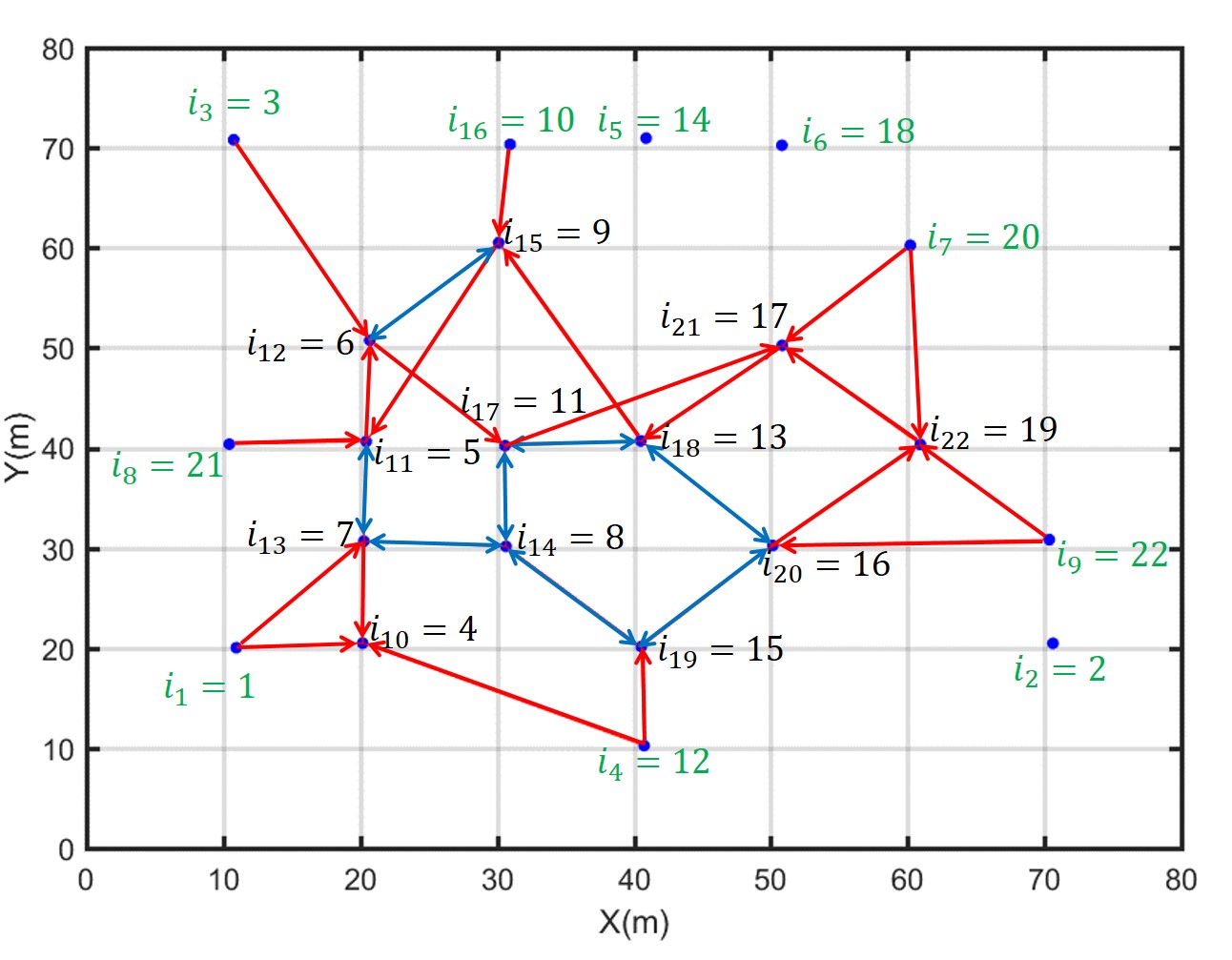}
\caption{Example reference formation used for collective motion simulation. A red arrow shows a unidirectional link to a follower from its in-neighbor agent. Blue arrows show bidirectional communication.}
\label{InitialFormation}
\end{figure}
\\

\subsubsection{{Classification of Agents as Leaders and Followers}} The node set $\mathcal{V}$ can be expressed as $\mathcal{V}=\mathcal{V}_B\bigcup\mathcal{V}_I$ where $\mathcal{V}_B=\{i_1,\cdots,i_{m_B}\}$ and $\mathcal{V}_I=\{i_{m_B+1},\cdots,i_{N}\}$ define boundary and interior agents, respectively.
 Given agents' reference positions, the following true statements are used to assign agent $i\in \mathcal{V}$ either as a leader or a follower:
\begin{enumerate}
\item{An $n$-D homogeneous deformation is defined by $n+1$ leaders \cite{rastgoftar2016continuum, rastgoftar2019multi}. Assuming leaders are selected from the boundary agents, 
    $\mathcal{V}_L\subset \mathcal{V}_B$. }
    \item{Non-leader boundary agents are the followers specified by $\left(\mathcal{V}_B\setminus \mathcal{V}_L\right)\subset \mathcal{V}_F$.}   
    \item{All interior agents are followers, thus, $\mathcal{V}_I\subset \mathcal{V}_F$.}
    \item{Agent $i$ is an interior agent and classified as a follower, if there exists a set of three agents $j_1$ $j_2$, and $j_{n+1}$ such that $\mathbf{\Lambda}\left(\mathbf{r}_{j_1,0},\mathbf{r}_{j_2,0},\mathbf{r}_{j_3,0},\mathbf{r}_{j_4,0}^n,\mathbf{r}_{i,0}^n\right)>0$, where $\mathbf{r}_{j_4,0}^n$ and $\mathbf{r}_{i,0}^n$ are assigned by Eq. \eqref{9aaa} when subscripts $i_1$, $i_2$, $i_3$, $i_4$, and $i_j$ are substituted by $j_1$, $j_2$, $j_3$, $j_4$, and $i$, respectively \cite{rastgoftar2019formal}.
    }
     \item{Assume $\mathbf{\Lambda}\left(\mathbf{r}_{j_1,0},\mathbf{r}_{j_2,0},\mathbf{r}_{j_3,0},\mathbf{r}_{j_4,0}^n,\mathbf{r}_{i,0}^n\right)$ has at least one negative entry for every $j_1,\cdots,j_{n+1}\in \mathcal{V}$ forming an $n$-D simplex, where $j_1\neq i$, $\cdots$, and $j_{n+1}\neq i$. Then, $i\in \mathcal{V}$ is a boundary agent \cite{rastgoftar2019formal}.}
    \item{If $i\in \mathcal{V}$ is not a follower agent, it is a boundary agent.}
    \item{Any $n+1$ boundary agents $j_1$, $\cdots$, $j_{n+1}$ can be selected as leaders. The remaining boundary agents are also considered as the followers.}
\end{enumerate}
To better clarify the above statements, the reference formation shown in Fig. \ref{InitialFormation} is considered. The vehicle team consists of $22$ agents ($\mathcal{V}=\{1,\cdots,22\}$). Set $\mathcal{V}_B=\{i_1,\cdots,i_{10}\}$ define the boundary agents, where $i_1=1$, $i_2=2$, $i_3=3$, $i_4=10$, $i_5=12$, $i_6=14$, $i_7=18$, $i_8=20$, $i_9=21$, $i_{10}=22$. While $\mathcal{V}_L=\{i_1,i_2,i_3\}$ specifies the leaders, $\mathcal{V}_B\subset \mathcal{V}_L$ defines the boundary followers. Boundary followers all communicate with leaders $1$, $2$, and $3$. Note that links from leaders $i_1$, $i_2$, and $i_3$ to boundary followers are not shown in Fig. \ref{InitialFormation}. Additionally, $\mathcal{V}_I=\{i_{11},\cdots,i_{22}\}$ defines interior agents, where $i_{11}=4$, $i_{12}=5$, $i_{13}=6$, $i_{14}=7$, $i_{15}=8$, $i_{16}=9$, $i_{17}=11$, $i_{18}=13$, $i_{19}=15$, $i_{20}=16$, $i_{21}=17$, and $i_{22}=19$ are the interior vehicles. Note that $\mathcal{V}_I\subset \mathcal{V}_F$ are all followers.

{\subsubsection{Followers' In-Neighbors, Communication Weights, and HDM Desired Trajectories}} 
The agent-tetrahedralization is used in this section to determine in-neighbor agents of interior followers in a homogeneous deformation coordination.
For every \textbf{\underline{interior} follower agent} $h\in\mathcal{V}_I\subset \mathcal{V}_F$, let
\begin{equation}
\resizebox{0.99\hsize}{!}{%
$\mathcal{F}_h=\left\{\left(j_1,\cdots,j_{n+1}\right)\in \underbrace{\mathcal{V}\times\cdots\times \mathcal{V}}_{n+1~\mathrm{times}}\bigg|\mathbf{\Lambda}\left(\mathbf{r}_{j_1,0},\cdots,\mathbf{r}_{j_{4},0}^n,\mathbf{r}_{h,0}^n\right)>\varrho_n\mathbf{1}_4\right\}$
}
\end{equation}
define admissible $n$-D simplexes enclosing interior follower $h$, where $\mathbf{1}_4\in \mathbb{R}^{4\times 1}$ is the one-entry vector and $\varrho>0$ is constant. 
{\color{black}
\begin{proposition}
Positive parameter $\varrho_n $ must be less than $\frac{1}{n+1}$ in an $n$-D homogeneous deformation ($n=2,3$).
\end{proposition}
\begin{proof}
For $n-D$ homogeneous transformation,
\[
\begin{split}
&\mathbf{1}_4^T\mathbf{\Lambda}\left(\mathbf{r}_{j_1,0},\cdots,\mathbf{r}_{j_{4},0}^n,\mathbf{r}_{h,0}^n\right)=\sum_{l=1}^4\lambda_l\left(\mathbf{r}_{j_1,0},\cdots,\mathbf{r}_{j_{4},0}^n,\mathbf{r}_{h,0}^n\right)=\\
&\sum_{l=1}^{n+1}\lambda_l\left(\mathbf{r}_{j_1,0},\cdots,\mathbf{r}_{j_{4},0}^n,\mathbf{r}_{h,0}^n\right)=1
\end{split}
\]
and 
\[
\varrho_n\leq \lambda_l\left(\mathbf{r}_{j_1,0},\cdots,\mathbf{r}_{j_{4},0}^n,\mathbf{r}_{h,0}^n\right)
\]
for $l=1,\cdots,n+1$, if $\left(j_1,j_2,j_3,j_4\right)\in \mathcal{F}_h$. Thus,
\[
(n+1)\varrho_n<\sum_{l=1}^{n+1}\lambda_l\left(\mathbf{r}_{j_1,0},\cdots,\mathbf{r}_{j_{4},0}^n,\mathbf{r}_{h,0}^n\right)=1
\]
which in turn implies that $\varrho_n<{\frac{1}{n+1}}$.
\end{proof}
}
In-neighbors of an interior follower $h\in \mathcal{V}_I$ is defined by set $\mathcal{N}_h=\{j_{1}^*,\cdots j_{n+1}^*\}$, where
\[
(j_1^*,\cdots,j_{n+1}^*)=\argmin\limits_{(j_1,\cdots,j_{n+1})\in \mathcal{F}_h}\sum_{k=1}^{n+1}\|\mathbf{r}_{j_k,0}-\mathbf{r}_{h,0}\|.
\]
In other words, the $n+1$ closest agents belonging to set $\mathcal{F}_h$ are considered as the in-neighbors of follower $h\in \mathcal{V}_F$. 

Every \textbf{\underline{boundary} follower agent} $j\in \mathcal{V}_B\setminus \mathcal{V}_L$ communicates with $n+1$ leaders defined by $\mathcal{V}_L$. Therefore, $\mathcal{N}_j=\mathcal{V}_L$ defines in-neighbor agent of vehicle $j\in \mathcal{V}_L\subset \mathcal{V}_B$. 

\textit{\underline{Followers' Communication Weights:}} Each communication weight $w_{i,j_k}$ ($k=1,\cdots,n+1$) is specified based on reference positions of follower vehicle $i\in \mathcal{V}_F$ and in-neighbor vehicle $j_k\in \mathcal{N}_i=\{j_1,\cdots,j_{n+1}\}$ as follows:
\begin{equation}
\label{ComWeights}
    \begin{bmatrix}
          w_{i,j_1}&\cdots&w_{i,j_4}
    \end{bmatrix}
    ^T
    =\mathbf{\Lambda}\left(\mathbf{r}_{j_1,0},\mathbf{r}_{j_2,0},\mathbf{r}_{j_3,0},\mathbf{r}_{j_4,0}^n,\mathbf{r}_{i,0}^n\right),
\end{equation}
where
\begin{subequations}
 \label{9aaab}
 \begin{equation}
     \mathbf{r}_{j_4,0}^n=
     \begin{cases}
     \mathbf{r}_{j_4,0}&n=3\\
     \mathbf{r}_{j_1,0}+\Xi\left(\mathbf{r}_{j_3,0}-\mathbf{r}_{j_1,0}\right)\times \left(\mathbf{r}_{j_2,0}-\mathbf{p}_{j_1,0}\right)&n=2
     \end{cases}
     ,
 \end{equation}
 \begin{equation}
     \mathbf{r}_{i,0}^n=
     \begin{cases}
     \mathbf{r}_{i,0}&n=3\\
      {\mathbf{r}}_{i,0}-\left({\mathbf{r}}_{i,0}\cdot\mathbf{n}_{1-4}\right)\mathbf{n}_{1-4}&n=2
     \end{cases}
     ,
 \end{equation}
\end{subequations}
and $\mathbf{n}_{1-4}=\mathbf{n}_{1-4}(\mathbf{r}_{j_1,0},\mathbf{r}_{j_2,0},\mathbf{r}_{j_3,0})$ is determined using Eq. \eqref{n144} when $n=2$.
Given followers' communication weights, the weight matrix $\mathbf{W}=[W_{jh}]\in \mathbb{R}^{(N-n-1)\times N}$ is defined as follows:
\begin{equation}
    W_{jh}=
    \begin{cases}
    w_{i_j,i_h}&i_h\in \mathcal{N}_j\wedge i_j\in \mathcal{V}_F\\
    0&\mathrm{otherwise}
    \end{cases}
\end{equation}
Matrix ${\mathbf{W}}$ can be partitioned as follows:
\begin{equation}
\label{Wpartition}
    {\mathbf{W}}=
    \left[
    \begin{array}{c|c}
         \mathbf{B}&\mathbf{A} 
    \end{array}
    \right],
\end{equation}
where $\mathbf{B}\in \mathbb{R}^{\left(N-n-1\right)\times\left(n+1\right)}$ and $\mathbf{A}\in \mathbb{R}^{\left(N-n-1\right)\times\left(N-n-1\right)}$ are non-negative matrices.

\underline{\textit{HDM Desired Trajectory:}} Local desired trajectory of agent $i\in \mathcal{V}$ is defined as follows:
\begin{equation}
\label{LOCALDESIRED}
    \mathbf{r}_{i,d}(t)=
    \begin{cases}
    \mathbf{r}_{i,c}&i\in \mathcal{V}_L\\
    \sum_{h\in \mathcal{N}_i}w_{i,h}\mathbf{r}_h&i\in \mathcal{V}_F
    \end{cases}
    .
\end{equation}
Note that global and local desired positions of leader agent $j\in \mathcal{V}_L$ are the same at any time $t$. The component $\mu\in\{x,y,z\}$ of the local desired positions of followers satisfy the following relation:
\begin{equation}
    \begin{split}
    \mu\in\{x,y,z\},~\forall t,\qquad \mathbf{P}_{\mu,d}^F(t)=&\mathbf{A}\mathbf{P}_{\mu}^F(t)+\mathbf{B}\mathbf{P}_{\mu}^L(t)\\
\end{split}
,
\end{equation}
where $\mathbf{A}$ and $\mathbf{B}$ were previously introduced in Eq. \eqref{Wpartition}. $\mathbf{P}_\mu^L=\left[\mu_{i_1}~\cdots~\mu_{i_{n+1}}\right]^T$, and $\mathbf{P}_\mu^F=\left[\mu_{i_{n+2}}~\cdots~\mu_{i_{N}}\right]^T$ assign the component $\mu\in \{x,y,z\}$ of actual positions of leaders and followers, respectively. Furthermore, $\mathbf{P}_{\mu,d}^F=\left[\mu_{i_{n+2},d}~\cdots~\mu_{i_{N},d}\right]^T$ assigns component $\mu\in \{x,y,z\}$ of the local desired positions for all followers. 


{\textbf{Key Property of Homogeneous Deformation:}} If followers' communication weights are consistent with agents' reference positions and obtained by Eq. \eqref{ComWeights}, then, the following relation is true:

 \begin{equation}
     \mathbf{W}_L=-\mathbf{D}^{-1}\mathbf{B}=
     \begin{bmatrix}
     \alpha_{i_{n+2},i_1}&\cdots&\alpha_{i_{n+2},i_{n+1}}\\
     \vdots&\ddots&\vdots\\
     \alpha_{i_N,i_1}&\cdots&\alpha_{i_N,i_{n+1}}\\
     \end{bmatrix}
 \end{equation}
  where 
 \[
\mathbf{D}=-\mathbf{I}+\mathbf{A}
 \]
 is Hurwitz (See the proof in Ref. \cite{rastgoftar2016continuum}). Let $\mathbf{P}_{\mu,c}^L=\left[\mu_{{i_1},c}~\cdots~\mu_{i_{n+1},c}\right]^T$ and $\mathbf{P}_{\mu,c}^F=\left[\mu_{{i_{n+2}},c}~\cdots~\mu_{i_{N},c}\right]^T$ specify component $\mu\in\{x,y,z\}$ of the global desired positions of leaders and followers, respectively.  Given the global desired position of followers defined by Eq. \eqref{GloBDesPos},   $\mathbf{P}_{\mu,c}^F$ is defined based on $\mathbf{P}_{\mu,c}^L$ by
 \begin{equation}
 \mu\in \{x,y,z\},~\forall t,\qquad     \mathbf{P}_{\mu,c}^F(t)=\mathbf{W}_L\mathbf{P}_{\mu,c}^L(t).
 \end{equation}
 \begin{lemma}\label{lemmma2}
 Every entry of matrix $\mathbf{D}^{-1}$ is non-positive.
 \end{lemma}
  \begin{proof}
  Diagonal entries of matrix $\mathbf{D}$ are all $-1$ while the off-diagonal entries of $\mathbf{D}$ are either $0$ or positive. Using the Gauss-Jordan elimination method,  the augmented matrix $\mathbf{D}_a=\left[\mathbf{D}\big|\mathbf{I}\right]\in \mathbb{R}^{\left(N-n-1\right)\times 2\left(N-n-1\right)}$ can be converted to matrix $\tilde{\mathbf{D}}_a=\left[\mathbf{I}\big|\mathbf{D}^{-1}\right]\in \mathbb{R}^{\left(N-n-1\right)\times 2\left(N-n-1\right)}$ only by performing row algebraic operations. Entries of the lower triangle of matrix $\mathbf{D}$ can be all converted to $0$, if  a top row is multiplied by a \underline{negative} scalar and the outcome is added to the other rows. Elements of the upper triangular submatrix of $\mathbf{L}$ can be similarly zeroed, if the bottom row is multiplied by a \underline{negative} scalar and the outcome is added to  the other rows.  Therefore, $\mathbf{D}^{-1}$, obtained by performing these row operations on $\mathbf{L}_a$,  is non-negative.  
  \end{proof}

 \begin{lemma}
 Define the
 \textbf{local-desired error} vector $\mathbf{E}_{\mu,d}^F=\left[\mu_{i_{n+2},d}-\mu_{i_{n+2}}~\cdots~\mu_{i_N,d}-\mu_{i_N}\right]^T$, and the \textbf{global-desired error} vectors $\mathbf{E}_{\mu,c}^L=\left[\mu_{i_1,c}-\mu_{i_{1}}~\cdots~\mu_{i_{n+1},c}-\mu_{i_{n+1},d}\right]^T$ and $\mathbf{E}_{\mu,c}^F=\left[\mu_{i_{n+2},c}-\mu_{i_{n+2}}~\cdots~\mu_{i_N,c}-\mu_{i_N}\right]^T$ where $\mu\in \{x,y,z\}$. The following relations are true:
 \begin{subequations}
 \begin{equation}
 \label{LocalErrorRelation}
   \mu\in\{x,y,z\},~\forall t,\qquad  \mathbf{E}_{\mu,d}^F(t)=\mathbf{D}\mathbf{P}_q^F(t)+\mathbf{B}\mathbf{P}_q^L(t),
 \end{equation}
 \begin{equation}
  \label{MainLocalErrorRelation}
 \mu\in\{x,y,z\},~\forall t,\qquad    \mathbf{E}_{q,c}^F=-\mathbf{D}^{-1}\mathbf{E}_{q,d}^F+\mathbf{B}\mathbf{E}_{q,c}^L,
 \end{equation}
 \end{subequations}
 \end{lemma}
\begin{proof}
Let
 \[
 \mu\in \{x,y,z\},i_j\in \mathcal{V}_F,\qquad \mu_{i_j,d}=\sum_{k\in \mathcal{N}_{i_j}}w_{i_j,k}{\mu}_k,
 \]
then, row $j$ of relation \eqref{LocalErrorRelation} can be expressed as follows:
 \[
   \mu_{i_{j+n+1},d}-\mu_{i_{j+n+1}}=-\mu_{i_{j+n+1}}+\sum_{k\in \mathcal{N}_{i_j}}w_{i_j,k}{\mu}_k,
 \]
where $\mu\in \{x,y,z\}$, $i_j\in \mathcal{V}_F$. Considering the key property of homogneous transformation, $\mathbf{B}=-\mathbf{D}\mathbf{W}_L$ and $\mathbf{P}_{\mu,c}^L$,  Eq. \eqref{LocalErrorRelation} can be rewritten as
 \[
 \begin{split}
\mathbf{E}_{\mu,d}^F=&\mathbf{D}\left(\mathbf{P}_\mu^F-\mathbf{W}_L\underbrace{\left(\mathbf{P}_{\mu,c}^L-\mathbf{E}_{\mu,c}^L\right)}_{\mathbf{P}_\mu^L}\right)=\mathbf{D}\left(-\mathbf{E}_{\mu,c}^F-\mathbf{B}\mathbf{E}_{\mu,c}^L\right).
\end{split}
 \]
where $\mu\in\{x,y\}$.  Therefore, $\mathbf{E}_{\mu,c}^F=-\mathbf{D}^{-1}\mathbf{E}_{\mu,d}^F+\mathbf{B}\mathbf{E}_{\mu,c}^L$.
 \\
\end{proof}

\begin{theorem}\label{theorem1}
 Assume control inputs $\mathbf{U}_L$ and $\mathbf{U}_F$ are designed so that
\begin{equation}
    \forall j\in \mathcal{V},~\forall \mu\in\left\{x,y,z\right\},\qquad \left|{\mu}_j-{\mu}_{j,d}\right|\leq \Delta_\mu,
\end{equation}
where $\mu_j$  and $\mu_{j,d}=\sum_{h\in \mathcal{N}_j}w_{j,h}\mu_h$ are components $\mu\in \{x,y,z\}$ of the \underline{actual} and  \underline{local desired} positions of vehicle $i$; communication weight $w_{j,h}$ is obtained using Eq. \eqref{ComWeights}. Then,
\begin{equation}
\label{importantineq}
\begin{split}
    &\sqrt{\left(x_j-x_{j,c}\right)^2+\left(y_j-y_{j,c}\right)^2+\left(y_j-y_{j,c}\right)^2}\leq\Delta,
\end{split}
\end{equation}
where 
\begin{subequations}
\begin{equation}
\label{DELTAAAAAAAAAAA}
\begin{split}
    \Delta=\Xi_{\mathrm{max}}\sqrt{\Delta_x^2+\Delta_y^2+\Delta_z^2},\\
\end{split}
\end{equation}
 \begin{equation}
\begin{split}
    \Xi_{\mathrm{max}}=
    &\max\limits_{l}\left(-\sum_{j=1}^{N-n-1}\mathbf{D}_{lj}^{-1}+\sum_{j=1}^{n+1}\mathbf{B}_{lj}\right)\\
\end{split}
\end{equation}
\end{subequations}

 \end{theorem}
 \begin{proof}
 Considering Eq. \eqref{MainLocalErrorRelation}, we can write
 \[
 \begin{split}
     &\left|\mu_{i_j,c}-\mu_{i_j}\right|=\bigg|-\sum_{j=1}^{N-n-1}D_{lj}^{-1}\left(\mu_{i_{j+n+1},c}-\mu_{i_{j+n+1}}\right)+\\
     &\sum_{j=1}^{n+1}B_{lj}\left(\mu_{i_j,c}-\mu_{i_j}\right)\bigg|\leq-\sum_{j=1}^{N-n-1}D_{lj}^{-1}\left|\mu_{i_{j+n+1},c}-\mu_{i_{j+n+1}}\right|+\\
     &\sum_{j=1}^{+1}B_{lj}\left|q_{i_j,c}-q_{i_j}\right|\leq-\sum_{j=1}^{N-n-1}D_{lj}^{-1}\Delta_\mu+\sum_{j=1}^{+1}B_{lj}\Delta_\mu\\
     \leq &\Delta_\mu\max\limits_{l}\left(-\sum_{j=1}^{N-n-1}\mathbf{D}_{lj}^{-1}+\sum_{j=1}^{n+1}\mathbf{B}_{lj}\right)=\Xi_{\mathrm{max}}\Delta_\mu\\
 \end{split}
 \]
 for $\mu\in\{x,y,z\}$. Therefore, inequality \eqref{importantineq} is satisfied.
 \end{proof}
 Theorem \ref{theorem1} specifies an upper limit for deviation of actual position of vehicle $i$ from the desired coordination defined at HDM. It is assumed that every vehicle is enclosed by a vertical cylinder of radius $\epsilon$, and $d_{min}$ denotes the minimum separation distance between every vehicle pair in the reference configuration. Then, inter-agent collision avoidance is guaranteed in a homogeneous deformation coordination, if the following inequality constraint is satisfied at any time $t$ \cite{rastgoftar2016continuum}:
 \begin{equation}
     \forall t, \qquad \min\{\sigma_1(t),\sigma_2(t),\sigma_3(t)\}\geq \dfrac{\Delta+\epsilon}{{\frac{d_{min}}{}2}+\epsilon}.
 \end{equation}






\subsubsection{HDM Control System}
It is assumed that vehicle $j\in \mathcal{V}$ has a nonlinear dynamics given by
\begin{equation}
\label{vehicledynamics}
\begin{cases}
\dot{\mathbf{x}}_j=\mathbf{f}_j(\mathbf{x}_j,\mathbf{u}_j)\\
\mathbf{r}_j=[x_j~y_j~z_j]^T
\end{cases}
,
\end{equation}
where $\mathbf{x}_j\in \mathbb{R}^{n_{\textbf{x},j}\times 1}$ and $\mathbf{u}_j\in \mathbb{R}^{n_{\textbf{u},j}\times 1}$ are the state and input vectors, and $\mathbf{r}_j=[x_j~y_j~z_j]^T$ is the actual position of vehicle $j$ considered as the output of vehicle $j$. 
As aforementioned, leaders move independently at the HDM. Therefore, $\mathcal{N}_i=\emptyset$, if $i\in \mathcal{V}_L$. 
Dynamics of the vehicle team is given by:
\[
\mathrm{Leaders:}
\qquad 
\begin{cases}
\dot{\mathbf{X}}_L
=\mathbf{F}_L\left({\mathbf{X}}_L,\mathbf{U}_L\right)
\\
\mathbf{R}_L=\mathrm{vec}\left(
\begin{bmatrix}
\mathbf{r}_{i_1}&
\cdots&
\mathbf{r}_{i_{n+1}}
\end{bmatrix}^T\right)
\end{cases}
\]
\[
\mathrm{Followers:}
\qquad 
\begin{cases}
\dot{\mathbf{X}}_F
=\mathbf{F}_F\left({\mathbf{X}}_F,\mathbf{U}_F\right)
\\
\mathbf{R}_F=\mathrm{vec}\left(
\begin{bmatrix}
\mathbf{r}_{i_{n+2}}&
\cdots&
\mathbf{r}_{i_{N}}
\end{bmatrix}^T\right)
\end{cases}
\]
where $\mathrm{vec}(\left[\cdot\right])$ vectorizes matrix $\left[\cdot\right]$, $\mathbf{X}_F$ and $\mathbf{X}_L$ are the state vectors representing leaders and followers, respectively, $\mathbf{U}_F$ and $\mathbf{U}_L$ are the leaders' and followers' control inputs. $\mathbf{F}_L=[\mathbf{f}_{i_1}^T~\cdots \mathbf{f}_{i_{n+1}}^T]^T$ and $\mathbf{F}_F=[\mathbf{f}_{i_{n+2}}^T~\cdots \mathbf{f}_{i_{N}}^T]^T$ are smooth functions where $\mathbf{f}_{i_k}$ ($k=1,\cdots,N_F$) specifies the dynamics of vehicle $i_k$ previously given in \eqref{vehicledynamics}. Also, $\mathbf{R}_L=\mathrm{vec}\left([\mathbf{r}_{i_1}~\cdots \mathbf{r}_{i_{N_L}}]^T\right)$ and $\mathbf{R}_F=\mathrm{vec}\left([\mathbf{r}_{i_{N_L+1}}~\cdots \mathbf{r}_{i_{N_F}}]^T\right)$ where $\mathbf{r}_{i_k}=[x_{i_k}~y_{i_k}~z_{i_k}]$ denotes actual position of vehicle $i_k$ ($k=1,\cdots,N_F$), $\mathbf{X}_L=\left[\mathbf{x}_{i_1}^T~\cdots~\mathbf{x}_{i_{n+1}}^T\right]^T$, $\mathbf{U}_L=\left[\mathbf{u}_{i_1}^T~\cdots~\mathbf{u}_{i_{n+1}}^T\right]^T$, $\mathbf{X}_F=\left[\mathbf{x}_{i_{n+2}}^T~\cdots~\mathbf{x}_{i_{N}}^T\right]^T$, and $\mathbf{U}_L=\left[\mathbf{u}_{i_{n+2}}^T~\cdots~\mathbf{u}_{i_{N}}^T\right]^T$ . 
\begin{figure}[ht]
\centering
\includegraphics[width=3.3 in]{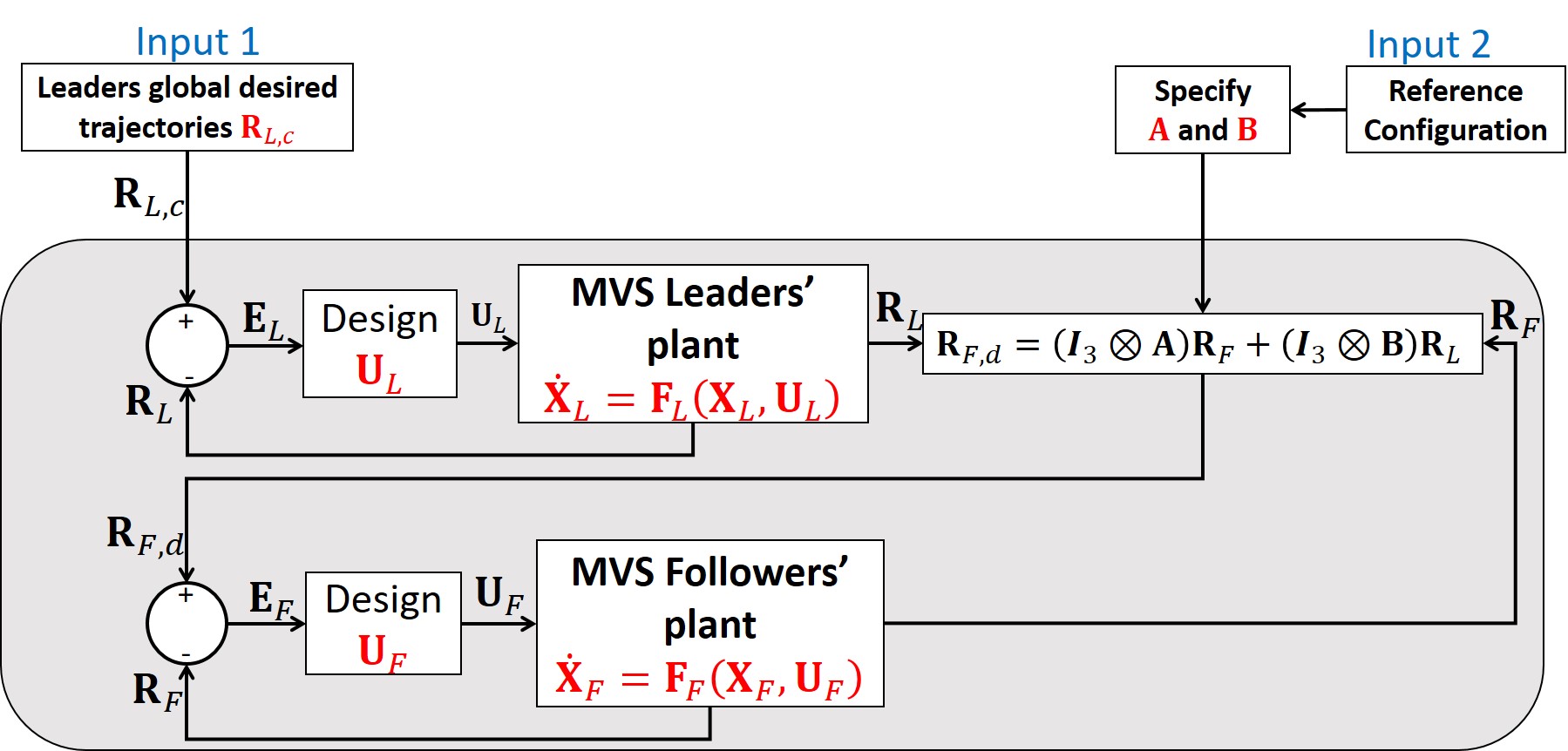}
\caption{Functionality of the cooperative team when HDM is active.  }
\label{HDMFunctionality}
\end{figure}
Fig. \ref{HDMFunctionality} shows the functionality of the cooperative control system in HDM. As shown the  system has the following inputs:
\begin{enumerate}
\item{Global desired trajectories of all leaders specified by vector $\mathbf{R}_L(t)$ at any time $t$.}
    \item{Matrix $\mathbf{A}$ and $\mathbf{B}$ assigned based on the cooperative team reference configuration using relation \eqref{Wpartition}.}
\end{enumerate}
Leader global desired trajectories can be safely planned so that collision with obstacles and inter-agent collision are both avoided while the leaders' distances between initial  and target states are minimized. Leader path planning using A* search and particle swarm optimization were previously studied in Refs. \cite{rastgoftar2019multi, liang2019multi}. Control inputs $\mathbf{U}_L$ and $\mathbf{U}_F$ can be assigned using existing approaches so the actual trajectory $\mathbf{r}_j$ is asymptotically tracked $\mathbf{r}_{j,c}$ for every vehicle $j\in \mathcal{V}$; specific analysis of tracking is beyond the scope of this paper. 
\subsection{Containment Exclusion Mode (CEM)}
\label{Containment Exclusion Mode (CEM)}
CEM is activated when there exists at least one  vehicle experiencing a failure or anomaly in containment domain $\Omega_{\mathrm{con}}$. Failed agent(s) are wrapped with an exclusion zone and healthy agents must be routed or "flow" around. Thus, $N_F(t)<N(t)$ and $\left|\mathcal{V}_A(t)\right|>0$ at any time $t$ when CEM is active. For  CEM, the coordinate transformation defined in \eqref{CoordinateTransformation} is used to assign the desired agent coordination. In particular,  potential function $\phi$ and stream function $\psi$ are determined by combining ``Uniform'' and ``Doublet`` flows:
\[
\begin{split}
    \phi(x,y,t)=&\phi_U(x,y,t)+\phi_D(x,y,t)\\
    \psi(x,y,t)=&\psi_U(x,y,t)+\psi_D(x,y,t)\\
\end{split}
\]
where the subscripts $U$ and $D$ are associated with ``Uniform'' and ``Doublet'', respectively. For the uniform flow pattern,
\begin{subequations}
\label{PHIU}
\begin{equation}
    \phi_U(x,y,t)=u_\infty(t)\left(x\cos\theta_\infty(t)+y\sin\theta_\infty(t)\right)
\end{equation}
\begin{equation}
    \psi_U(x,y,t)=u_\infty(t)\left(-x\sin\theta_\infty(t)+y\cos\theta_\infty(t)\right),
\end{equation}
\end{subequations}
define the potential and stream fields, respectively, where $u_\infty(t)$ and $\theta_\infty(t)$ are design parameters. Furthermore,
\[
\phi_D=\sum_{i\in\mathcal{V}_A}\phi_{D,i}\qquad\mathrm{and}\qquad \psi_D=\sum_{i\in\mathcal{V}_A}\psi_{D,i},
\]
define potential and stream fields of the Doublet flow, respectively, where
\begin{subequations}
\label{phipsi}
\begin{equation}
\phi_{D,i}=\dfrac{\delta_{i}(t)\left[\cos\gamma_{i}(t)\left(x-a_{i}(t)\right)+\sin\gamma_{i}(t)\left(y-b_{i}(t)\right)\right]}{\left(x-a_{i}(t)\right)^2+\left(y-b_{i}(t)\right)^2},
\end{equation}
\begin{equation}
\psi_{D,i}=\dfrac{\delta_{i}(t)\left[-\sin\gamma_{i}(t)\left(x-a_{i}(t)\right)+\cos\gamma_{i}(t)\left(y-b_{i}(t)\right)\right]}{\left(x-a_{i}(t)\right)^2+\left(y-b_{i}(t)\right)^2},
\end{equation}
\end{subequations}
and ${\Delta}_{i}$, ${a}_{i}$, ${b}_{i}$ are design parameters specifying the geometry and location of anomalous/failed agent $i\in \mathcal{V}_A$ in the motion space.
By treating agent coordination as ideal fluid flow, we can exclude failed agent $i\in \mathcal{V}_A$ by wrapping them with a closed surface $\psi(x,y,t)=\psi_{i,0}$, where $\psi_{i,0}$ is constant. Furthermore, healthy vehicle  $j\in \mathcal{V}_H$ moves along the global desired trajectories
\begin{equation}
\label{PlatooningCondition}
 \psi_{j,c}(t)=\psi(x_{j,c}(t),y_{j,c}(t),t)=\psi_{j,0}~(\mathrm{constant}).
\end{equation}
{\color{black}where $\psi_{j,0}$ is assigned based on position of vehicle $j\in \mathcal{V}_{H}$ at the time the cooperative team enters the CEM.} 


\begin{theorem}\label{theorem222}
Suppose $\mathbf{J}\left(x_j,y_j\right)$ is the Jacobian matrix defined by \eqref{Jaaaacobian}, and the desired trajectory of every agent $j\in \mathcal{V}_H$ satisfies Eq. \eqref{StickingCondition}. Define
\begin{subequations}
\begin{equation}
    \mathbf{H}_{s,j}=-\dfrac{1}{\left|\mathbf{J}\left(x_j,y_j\right)\right|}
    \begin{bmatrix}
    \dfrac{\partial \psi}{\partial y_{j,c}}\\-\dfrac{\partial \psi}{\partial x_{j,c}}
    \end{bmatrix}
    \begin{bmatrix}
    \dfrac{\partial \phi}{\partial u_{\infty}}&\dfrac{\partial \phi}{\partial\theta_{\infty}}
    \end{bmatrix}
\end{equation}
\begin{equation}
\resizebox{0.99\hsize}{!}{%
    $\mathbf{H}_{a,j,i_l}=-\dfrac{1}{\left|\mathbf{J}\left(x_j,y_j\right)\right|}
    \begin{bmatrix}
    \dfrac{\partial \psi}{\partial y_{j,c}}\\-\dfrac{\partial \psi}{\partial x_{j,c}}
    \end{bmatrix}
    \begin{bmatrix}
    \dfrac{\partial \phi}{\partial a_{i_l}}&\dfrac{\partial \phi}{\partial b_{i_l}}&\dfrac{\partial \phi}{\partial \Delta_{i_l}}
    \end{bmatrix}
    ,\qquad l=N_F+1,\cdots,N,$
    }
\end{equation}
\begin{equation}
    \mathbf{H}_{T,j}=\dfrac{1}{\left|\mathbf{J}\left(x_j,y_j\right)\right|}
    \begin{bmatrix}
    \dfrac{\partial \psi}{\partial y_{j,c}}\\-\dfrac{\partial \psi}{\partial x_{j,c}}
    \end{bmatrix}
    ,
\end{equation}
 \begin{equation}
    \dot{\mathbf{q}}_{c,\mathrm{CEM}}=
 \begin{bmatrix}
       \dot{u}_{\infty}~\dot{\theta}_{\infty}
 \end{bmatrix}
 ^T
 ,
\end{equation}
 \begin{equation}
    \dot{\mathbf{q}}_{u,\mathrm{CEM}}=
 \begin{bmatrix}
       \dot{a}_{i_{N_F+1}}~\dot{b}_{i_{N_F+1}}~\dot{\Delta}_{i_{N_F+1}}~\cdots~\dot{a}_{i_{N}}~\dot{b}_{i_{N}}~\dot{\Delta}_{i_{N}}
 \end{bmatrix}
 ^T
 ,
\end{equation}
\end{subequations}
Then, the CEM global desired trajectory can be defined by Eq. \eqref{MAINNNNN}, where $\gamma=\mathrm{CEM}$,
\begin{subequations}
\label{IdealFluidFlowDesired}
\begin{equation}
    \dot{\mathbf{q}}_{\mathrm{CEM}}=
 \begin{bmatrix}
       \dot{\mathbf{q}}_{c,\mathrm{CEM}}\\
       \dot{\mathbf{q}}_{u,\mathrm{CEM}}\\
       v_{\phi}
 \end{bmatrix}
 \in \mathbb{R}^{3+3\left(N-N_F\right)\times 1}
 ,
\end{equation}
\begin{equation}
    \mathbf{H}_{j,\mathrm{CEM}}=
 \begin{bmatrix}
       1&0\\
       0&1\\
       \frac{\partial z_{j,c}}{\partial x_{j,c}} & \frac{\partial z_{j,c}}{ \partial y_{j,c}}
 \end{bmatrix}
 \begin{bmatrix}
       \mathbf{H}_s&\mathbf{H}_{a,i_{N_F+1}}&\cdots&\mathbf{H}_{a,i_{N}}&\mathbf{H}_{T,j}
 \end{bmatrix}
\end{equation}
\end{subequations}
for every agent $j\in \mathcal{V}_H$, where $v_\phi=\frac{\partial \phi}{ \partial t}$ is the desired sliding speed of healthy vehicles along their desired stream lines. 
\end{theorem}
\begin{proof}
Per the prescribed CEM protocol vehicle $i$ slides along the stream line $\psi_{j,c}=\psi_{j,0}$ at any time $t$. Eq. \eqref{StickingCondition} must be satisfied at every point $(x_{j,c},y_{j,c})$ and any time $t$. Given the sliding speed $v_{\phi}$, the following relation holds:
\begin{subequations}
\begin{equation}
    \begin{split}
    &\dfrac{\partial \phi}{\partial x_{j,c}}\dot{x}_{j,c}+\dfrac{\partial \phi}{\partial y_{j,c}}\dot{y}_{j,c}=- \dfrac{\partial \phi}{\partial u_{\infty}}\dot{u}_{\infty}- \dfrac{\partial \phi}{\partial \theta_{\infty}}\dot{\theta}_{\infty}\\
    &-\sum_{i\in \mathcal{V}_A}\left(\dfrac{\partial \phi}{\partial a_i}\dot{a}_i+\dfrac{\partial \phi}{\partial b_i}\dot{b}_i+\dfrac{\partial \phi}{\partial \Delta_i}\dot{\Delta}_i\right)+v_{\phi},\\
\end{split}
\end{equation}
\begin{equation}
    \begin{split}
    &\dfrac{\partial \psi}{\partial x_{j,c}}\dot{x}_{j,c}+\dfrac{\partial \psi}{\partial y_{j,c}}\dot{y}_{j,c}=0.
\end{split}
\end{equation}
\end{subequations}
Therefore, $x$ and $y$ components of agent $j\in \mathcal{V}_H$  global desired trajectory are updated by \eqref{MAINNNNN}, where $\mathbf{H}_{j,\mathrm{CEM}}$ and $\dot{\mathbf{q}}_{\mathrm{CEM}}$ are given by Eq. \eqref{IdealFluidFlowDesired} for agent $j\in \mathcal{V}_H$ at any time $t$.
\end{proof}
Design parameters $\dot{u}_\infty$, $\dot{\theta}_{\infty}$, $\dot{\Delta}_{i}$, $\dot{a}_{i}$, $\dot{b}_{i}$ ($i\in \mathcal{V}_A$), and $v_{\phi}$, {\color{black}obtained by taking time derivative from the generalized coordinates,} define group desired coordination for CEM. Note that $u_{\infty}$ and $\theta_{\infty}$ can be designed so that the ideal fluid flow coordination is optimized. However, the remaining design parameters are uncontrolled. 
{\color{black}
\begin{remark}
In general, design parameters 
 $\dot{u}_\infty$, $\dot{\theta}_{\infty}$, $\dot{\Delta}_{i}$, $\dot{a}_{i}$, $\dot{b}_{i}$ ($i\in \mathcal{V}_A$) can vary with time. 
However, this paper concentrates only on the steady-state CEM which will be achieved when $\dot{u}_\infty$, $\dot{\theta}_{\infty}$, $\dot{\Delta}_{i}$, $\dot{a}_{i}$, $\dot{b}_{i}$ ($i\in \mathcal{V}_A$) are all zeros. 
Therefore, potential and stream functions are defined by Eqs. \eqref{PHIU}, and \eqref{phipsi} simplifies to
\[
j\in \mathcal{V}_H,\qquad 
\begin{bmatrix}
    \dot{x}_{j,c}\\
    \dot{y}_{j,c}\\
    \dot{z}_{j,c}\\
    \end{bmatrix}
    =\begin{bmatrix}
       1&0\\
       0&1\\
       \frac{\partial z_{j,c}}{ \partial x_{j,c}}& \frac{\partial z_{j,c}}{ \partial y_{j,c}}
 \end{bmatrix}
   \mathbf{H}_{T,j}v_{\phi}.
\]
This requires an assumption for this work that the failed vehicle $i\in \mathcal{V}_A$ remains inside a predictable closed domain, with time-invariant geometry, until the time the failed agent is no longer in containment domain $\Omega_{\mathrm{con}}$ defined per Eq. \eqref{containmentdoamin}.

\end{remark}
}

\section{Continuum Deformation Anomaly Management}\label{Continuum Deformation Anomaly Management}
This section develops a hybrid model to manage transitions between CEM and HDM.
Section \ref{Anomaly Detection} develops a distributed   approach to detect a vehicle failure/anomaly followed by a supervisory control transition approach described in Section \ref{Supervisory Control}.

\subsection{Anomaly Detection}
\label{Anomaly Detection}
In this sub-section, we present a distributed model to detect situations in which agents have failed or are no longer cooperative. We then consider these agents anomalous or failed and add them to anomalous agent set $\mathcal{V}_A$. 

Consider an $n$-D homogeneous deformation where follower $i$ knows its own position and positions of  in-neighbor agents $\mathcal{N}_i=\{j_1,\cdots j_{n+1}\}$ at any time $t$. Let actual position  $\mathbf{r}_i(t)$ be expressed as the convex combination of agent $i$'s in-neighbors by 
\begin{subequations}
\begin{equation}
    i\in \mathcal{V}_F,\qquad \mathbf{r}_i(t)=\sum_{k=1}^{n+1}\varpi_{i,j_k}(t)\mathbf{r}_{j_k}.
\end{equation}
\begin{equation}
    i\in \mathcal{V}_F,\qquad \sum_{k=1}^{n+1}\varpi_{i,j_k}(t)=1,
\end{equation}
\end{subequations}
{\color{black}where $\varpi_{i,j_1}$ through $\varpi_{i,j_{n+1}}$ are called \textit{transient weights.}} If $\varkappa_n\left(\mathbf{r}_{j_1}(t),\cdots,\mathbf{r}_{j_{n+1}}(t)\right)=n$, transient weights $\varpi_{i,j_1}$ through $\varpi_{i,j_{n+1}}$ can be assigned based on agents' actual positions as follows:
\begin{equation}
\label{varpiiiiiiiii}
    \begin{bmatrix}
          \varpi_{i,j_1}\left(t\right)&\cdots&\varpi_{i,j_4}\left(t\right)
    \end{bmatrix}
    ^T
    =\mathbf{\Lambda}\left(\mathbf{r}_{j_1},\mathbf{r}_{j_2},\mathbf{r}_{j_3},\mathbf{r}_{j_4}^n,\mathbf{r}_{i}^n\right),
\end{equation}
where
\begin{subequations}
 \label{9aaac}
 \begin{equation}
     \mathbf{r}_{j_4}^n\left(t\right)=
     \begin{cases}
     \mathbf{r}_{j_4}\left(t\right)&n=3\\
     \mathbf{r}_{j_1}+\Xi\left(\mathbf{r}_{j_3,0}-\mathbf{r}_{j_1,0}\right)\times \left(\mathbf{r}_{j_2,0}-\mathbf{p}_{j_1,0}\right)&n=2
     \end{cases}
     ,
 \end{equation}
 \begin{equation}
     \mathbf{r}_{i,0}^n\left(t\right)=
     \begin{cases}
     \mathbf{r}_{i}&n=3\\
      {\mathbf{r}}_{i}-\left({\mathbf{r}}_{i}\cdot\mathbf{n}_{1-4}\right)\mathbf{n}_{1-4}&n=2
     \end{cases}
     ,
 \end{equation}
\end{subequations}
\begin{equation}
\label{rankkkkkkkkkkkkkkkk}
\Gamma_3\left(\mathbf{r}_{j_1}(t),\mathbf{r}_{j_2}(t),\mathbf{r}_{j_3}(t),\mathbf{r}_{j_4}(t),\right)=3.
\end{equation}
and $\mathbf{n}_{1-4}=\mathbf{n}_{1-4}(\mathbf{r}_{j_1},\mathbf{r}_{j_2},\mathbf{r}_{j_3})$ is determined based on agents' actual positions using Eq. \eqref{n144} when $n=2$.
\\
{\color{black}
\textbf{{Geometric Interpretation of Transient Weights}}: Let $d_{i,j_2,j_3}(t)$, $d_{i,j_3,j_1}(t)$, and $d_{i,j_1,j_2}(t)$ denote distances of point $i$ from the triangle sides $j_2-j_3$, $j_3-j_1$, and $j_1-j_2$, respectively. Assume $l_{j_1,j_2,j_3}(t)$, $l_{j_2,j_3,j_1}(t)$, $l_{j_3,j_1,j_2}(t)$ determine distances of vertices $j_1$, $j_2$, and $j_3$ from the triangle sides $j_2-j_3$,  $j_3-j_1$, $j_1-j_2$, respectively. Then, 
\begin{subequations}
 \begin{equation}
     \varpi_{i,j_1}(t)=\dfrac{d_{i,j_2,j_3}}{l_{j_1,j_2,j_3}},
 \end{equation}
 \begin{equation}
     \varpi_{i,j_2}(t)=\dfrac{d_{i,j_3,j_1}}{l_{j_2,j_3,j_1}},
 \end{equation}
 \begin{equation}
   \varpi_{i,j_3}(t)=\dfrac{d_{i,j_1,j_2}}{l_{j_3,j_1,j_2}}.
 \end{equation}
\end{subequations}

    
    
Geometric representations of $d_{i,j_2j_3}(t)$ and $l_{j_1,j_2j_3}(t)$ are shown in Fig.  \ref{PROOFMAIN} (a) when $n=2$.

For $n=3$, $d_{i,j_2,j_3,j_4}(t)$, $d_{i,j_3,j_4j_1}(t)$, $d_{i,j_4,j_1,j_2}(t)$, and $d_{i,j_1,j_2,j_3}(t)$ denote distance of point $i$ from the triangular surfaces $j_2-j_3-j_4$, $j_3-j_4-j_1$, $j_4-j_1-j_2$, $j_1-j_2-j_3$, respectively. Assume $l_{j_1,j_2,j_3,j_4}(t)$, $l_{j_2,j_3,j_4,j_1}(t)$, $l_{j_3,j_4,j_1,j_2}(t)$, and $l_{j_4,j_1,j_2,j_3}$ determine distance of vertices $j_1$, $j_2$, $j_3$, and $j_4$, from the triangular surfaces $j_2-j_3-j_4$, $j_3-j_4-j_1$, $j_4-j_1-j_2$, and $j_1-j_2-j_3$ respectively. Then,
\begin{subequations}
 \begin{equation}
     \varpi_{i,j_1}(t)=\dfrac{d_{i,j_2,j_3,j_4}}{l_{j_1,j_2,j_3,j_4}},
 \end{equation}
 \begin{equation}
     \varpi_{i,j_2}(t)=\dfrac{d_{i,j_3,j_4,j_1}(t)}{l_{j_2,j_3,j_4,j_1}(t)},
 \end{equation}
 \begin{equation}
     \varpi_{i,j_3}(t)=\dfrac{d_{i,j_4,j_1,j_2}(t)}{l_{j_3,j_4,j_1,j_1}(t)},
 \end{equation}
  \begin{equation}
    \varpi_{i,j_4}(t)=\dfrac{d_{i,j_4,j_1,j_2}(t)}{l_{j_3,j_4,j_1,j_1}(t)}.
 \end{equation}
\end{subequations}
    }    

\begin{theorem}\label{vapircomweight}
Assume HDM collective motion is guided by $n+1$ leaders, defined by $\mathcal{V}_L$, every follower $i\in \mathcal{V}_F$, communicates with $n+1$ in-neighbor agents, defined by $\mathcal{N}_i=\{j_1,\cdots,j_{n+1}\}$, where follower $i$'s in-neighbors form an $n$-D simplex at time $t$. If deviation of every agent $i\in \mathcal{V}$ from the global desired position $\mathbf{r}_{i,c}$ is less than $\Delta$ at time $t$ ($\left|\mathbf{r}_i(t)-\mathbf{r}_{i,c}(t)\right|\leq \Delta$), then followers' communication weights satisfy the following inequality:
    \begin{equation}
       \varpi_{i,j_1}^{\mathrm{min}}(t)\leq  w_{i,j_1}\leq \varpi_{i,j_1}^{\mathrm{max}}(t),
    \end{equation}
    where $w_{i,j}$ is constant communication wight of follower $i\in \mathcal{V}$ with in-neighbor $j_1$ assigned by Eq. \eqref{ComWeights}, and
    \begin{subequations}
    \begin{equation}
        \varpi_{i,j_1}^{\mathrm{min}}(t)=
        \begin{cases}
        \dfrac{d_{i,j_2,j_3}(t)-\Delta}{l_{j_1,j_2,j_3}(t)+2\Delta}&n=2\\
        \dfrac{d_{i,j_2,j_3,j_4}(t)-\Delta}{l_{j_1,j_2,j_3,j_4}+2\Delta}&n=3\\
        \end{cases}
        ,
    \end{equation}
    \begin{equation}
        \varpi_{i,j_1}^{\mathrm{max}}(t)=
        \begin{cases}
        \dfrac{d_{i,j_2,j_3,j_4}(t)+2\Delta}{l_{j_1,j_2,j_3}(t)-\Delta}&n=2\\
        \dfrac{d_{i,j_2,j_3,j_4}(t)+2\Delta}{l_{j_1,j_2,j_3,j_4}-\Delta}&n=3\\
        \end{cases}
        .
    \end{equation}
    \end{subequations}
\end{theorem}
specify lower and upper bounds for transient weight $\varpi_{i,j_1}(t)$ at time $t$.
\begin{proof}
If $\mathbf{\mathbf{r}}_i(t)=\mathbf{r}_{i,c}(t)$ for every agent $i\in \mathcal{V}$ at any time $t$, then $\varpi_{i,j_k}(t)=w_{i,j_k}$ ($k=1,\cdots,n+1$,~$j_k\in \mathcal{N}_i$). For $n=2$, we define a desired triangle $j_1-j_2-j_3$ with vertices placed at $\mathbf{r}_{j_1,c}$, $\mathbf{r}_{j_2,c}$, and $\mathbf{r}_{j_3,c}$. $D_{i,j_2,j_3}$ denotes the distance between the global desired position of agent $i$ and the triangle side $j_2,j_3$ while $L_{j_1,j_2,j_3}$ denotes the distance between the desired position of agent $j_1$ and the side $j_2-j_3$ of the desired tetrahedron. We also define an ``actual'' triangle with vertices positioned at $\mathbf{r}_{j_1}$, $\mathbf{r}_{j_2}$, and $\mathbf{r}_{j_3}$. When $\|\mathbf{r}_i(t)-\mathbf{r}_{i,c}(t)\|\leq \Delta $ is satisfied for agent $i\in \mathcal{V}$, then, 
\begin{subequations}
 \begin{equation}
     d_{i,j_2,j_3}(t)-2\Delta\leq D_{i,j_2,j_3}(t)\leq d_{i,j_2,j_3}(t)+2\Delta,
 \end{equation}
 \begin{equation}
     l_{i,j_2,j_3}(t)-2\Delta\leq L_{i,j_2,j_3}(t)\leq l_{i,j_2,j_3}(t)+2\Delta.
 \end{equation}
\end{subequations}
Therefore, $w_{i,j_1}=\dfrac{D_{i,j_2,j_3}(t)}{L_{i,j_2,j_3}(t)}\in \left[\dfrac{d_{i,j_2,j_3}-2\Delta}{L_{i,j_2,j_3}+2\Delta},\dfrac{d_{i,j_2,j_3}+2\Delta}{L_{i,j_2,j_3}-2\Delta}\right]$ (See Fig. \ref{PROOFMAIN}).  For $n=3$, vertices of the desired tetrahedron $j_1-j_2-j_3-j_4$ are placed at $\mathbf{r}_{j_1,c}$, $\mathbf{r}_{j_2,c}$, $\mathbf{r}_{j_3,c}$, and $\mathbf{r}_{j_4,c}$; vertices of the ``actual'' tetrahedron are positioned at $\mathbf{r}_{j_1}$, $\mathbf{r}_{j_2}$,  $\mathbf{r}_{j_3}$, and $\mathbf{r}_{j_3}$.   $D_{i,j_1,j_2,j_3}$ denotes the distance between the global desired position of agent $i$ and the tetrahedron surface $j_2,j_3,j_4$. $L_{j_1,j_2,j_3}$ denotes the distance between the desired position of agent $j$ and the surface $j_2-j_3-j_4$ of the desired tetrahedron. 
Assuming every agent $i\in \mathcal{V}$ satisfies safety constraint \eqref{importantineq}, 
\begin{subequations}
 \begin{equation}
     d_{i,j_2,j_3,j_4}(t)-2\Delta\leq D_{i,j_2,j_3,j_4}(t)\leq d_{i,j_2,j_3,j_4}(t)+2\Delta,
 \end{equation}
 \begin{equation}
     l_{i,j_2,j_3,j_4}(t)-2\Delta\leq L_{i,j_2,j_3,j_4}(t)\leq l_{i,j_2,j_3,j_4}(t)+2\Delta,
 \end{equation}
\end{subequations}
Therefore, 
\[
w_{i,j_1}=\dfrac{D_{i,j_2,j_3,j_4}(t)}{L_{i,j_2,j_3,j_4}(t)}\in \left[\dfrac{d_{i,j_2,j_3,j_4}-2\Delta}{L_{i,j_2,j_3,j_4}+2\Delta},\dfrac{d_{i,j_2,j_3,j_4}+2\Delta}{L_{i,j_2,j_3,j_4}-2\Delta}\right].
\]
\begin{figure}[ht]
\centering
\subfigure[]{\includegraphics[width=0.8\linewidth]{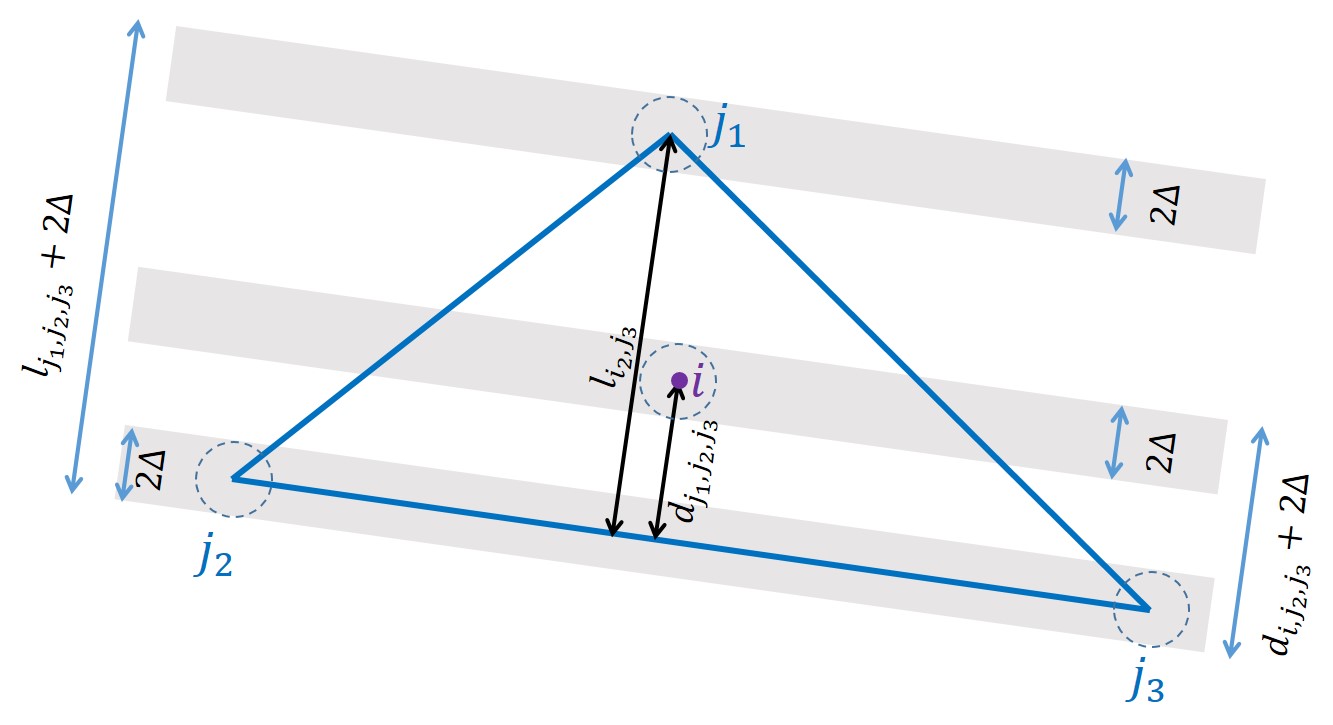}}
\subfigure[]{\includegraphics[width=0.8\linewidth]{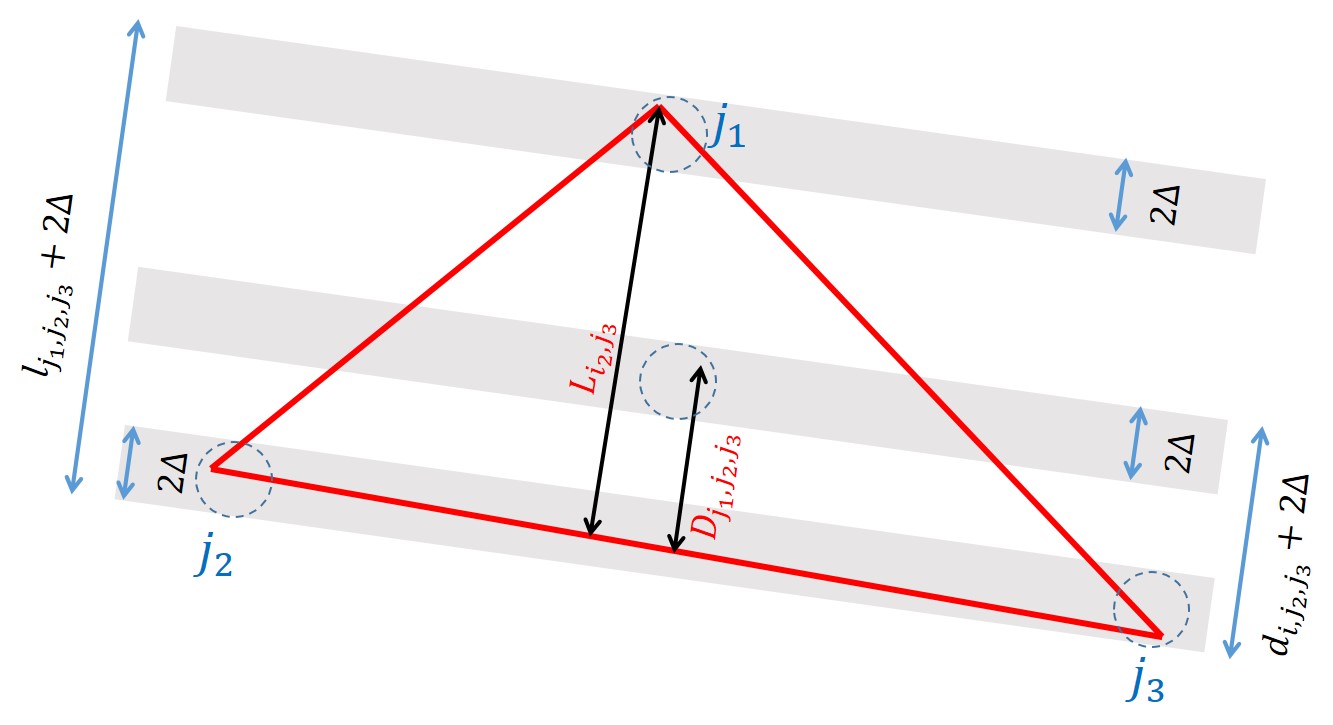}}
\vspace{-.4cm}
\caption{(a) ``Actual'' triangle constructed by the actual positions of  agents $j_1$, $j_2$, and $j_3$ at time $t$. (b) Desired triangle given by the global desired positions of agents $j_1$, $j_2$, and $j_3$ at time $t$.
}
 \label{PROOFMAIN}
\end{figure}
\end{proof}
Theorem \ref{vapircomweight} implies that HDM mode is active only if the following condition is satisfied:
\begin{equation}
    \label{HDMCondition}
    \forall i\in \mathcal{V},~k=1,\cdots,n+1,\qquad 
    \varpi_{i,j_k}^{\mathrm{min}}(t)\leq  w_{i,j_k}\leq \varpi_{i,j_k}^{\mathrm{max}}(t).\tag{$\Psi_{i,j_k}$},
\end{equation}
where $\mathcal{N}_i=\{j_1,\cdots,j_{n+1}\}$ defines in-neighbors of agent $i\in \mathcal{V}$. Therefore, if $\bigwedge_{i=1}^{N(t)}\bigwedge_{k=1}^{n+1}\Psi_{i,j_k}$ is satisfied at time $t$, HDM is active. {\color{black}Otherwise, an anomaly is detected.} Additionally, disjoint sets $\mathcal{V}_H$ and $\mathcal{V}_A$ are defined as follows:
\begin{subequations}
 \begin{equation}
     \mathcal{V}_{H}(t)=\left\{j\in \mathcal{V}(t)\bigg|\bigwedge_{h\in \mathcal{N}_h}\Psi_{i,j_k}~\mathrm{is~satisfied.}\right\},
 \end{equation}
 \begin{equation}
     \mathcal{V}_{A}(t)=\mathcal{V}(t)\setminus \mathcal{V}_{H}(t).
 \end{equation}
\end{subequations}

\begin{figure}[ht]
\centering
\includegraphics[width=3.3 in]{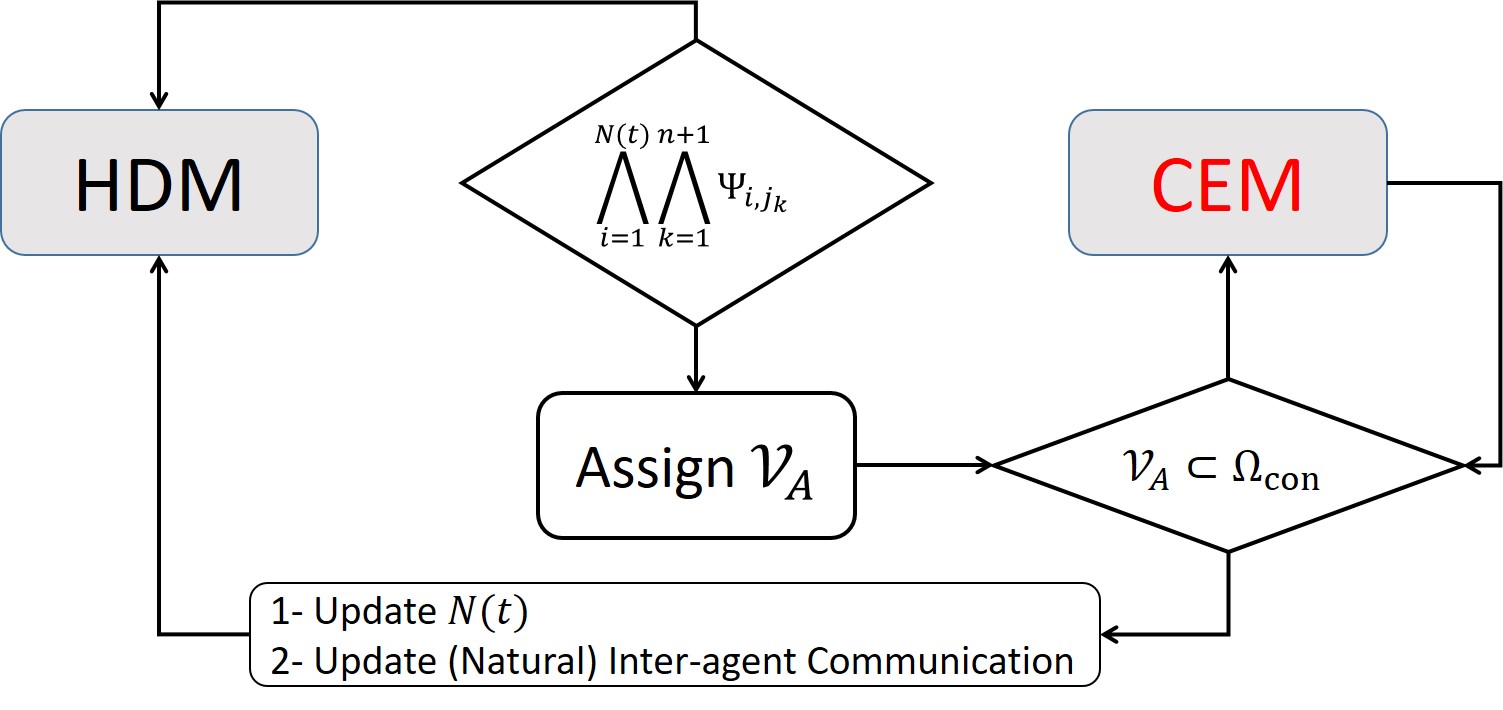}
\caption{Failed vehicle assignment and management by cooperative team leaders.}
\label{SupervisoryControl}
\end{figure}


\subsection{Vehicle Anomaly/Failure Management}
\label{Supervisory Control}
The Fig. \ref{SupervisoryControl} flowchart illustrates how vehicle failure can be managed by transition between ``HDM'' and ``CEM''. The following procedure is proposed:
\begin{enumerate}
    \item{Define containment domain $\Omega_{\mathrm{con}}\left(\mathbf{r},\mathbf{r}_{\mathrm{con}}(t)\right)$ using Eq. \eqref{containmentdoamin}.}
    \item{If there exists at least one failed agent inside the containment domain $\Omega_{\mathrm{con}}\left(\mathbf{r},\mathbf{r}_{\mathrm{con}}(t)\right)$, then 
    \[
    \bigwedge_{i=1}^{N(t)}\bigwedge_{k=1}^{n+1}\Psi_{i,j_k}
    \]
    is \underline{\textbf{not}} satisfied and CEM is activated.}
    \item{If agents contained by $\Omega_{\mathrm{con}}\left(\mathbf{r},\mathbf{r}_{\mathrm{con}}(t)\right)$ are all healthy, then $\bigwedge_{i=1}^{N(t)}\bigwedge_{k=1}^{n+1}\Psi_{i,j_k}$ is \textbf{satisfied} which in turn implies that $\mathcal{V}_A=\emptyset$ and HDM is active.}
\end{enumerate}

\section{Simulation Results}
\label{results}
 Consider collective motion in a $2$-D plane with invariant $z$ components for all agents at all times $t$. Suppose a multi-agent team consisting of $22$ vehicles is deployed with the initial formation shown in Fig. \ref{InitialFormation}. Given global desired positions of all agents at time $t$, the containment domain $\Omega_{\mathrm{com}}$ is defined for this case study as:
\[
\Omega_{\mathrm{con}}=\|\mathbf{r}-\mathbf{r}_{\mathrm{con}}\|_1\leq 40,
\]
where $\mathbf{r}_{\mathrm{con}}=\dfrac{1}{N(t)}\sum_{i\in \mathcal{V}_H}$ and $\|\cdot\|_1$ denotes the $1$-norm. Therefore, $\Omega_{\mathrm{con}}$ is a box with side length $80m$.

Without loss of generality, assume that every agent is a single integrator. The position of each agent $i$ is updated by
\begin{equation}
    i\in \mathcal{V},\qquad \dot{\mathbf{r}}_i=g(\mathbf{r}_{i,d}-\mathbf{r}_{i}),
\end{equation}
where $g=25$ is constant, $\mathbf{r}_i$ is the actual position of agent $i$, and local desired position $\mathbf{r}_{i,d}$ was defined in Eq. \eqref{LOCALDESIRED}. 
\begin{figure}[ht]
\centering
\subfigure[]{\includegraphics[width=0.9\linewidth]{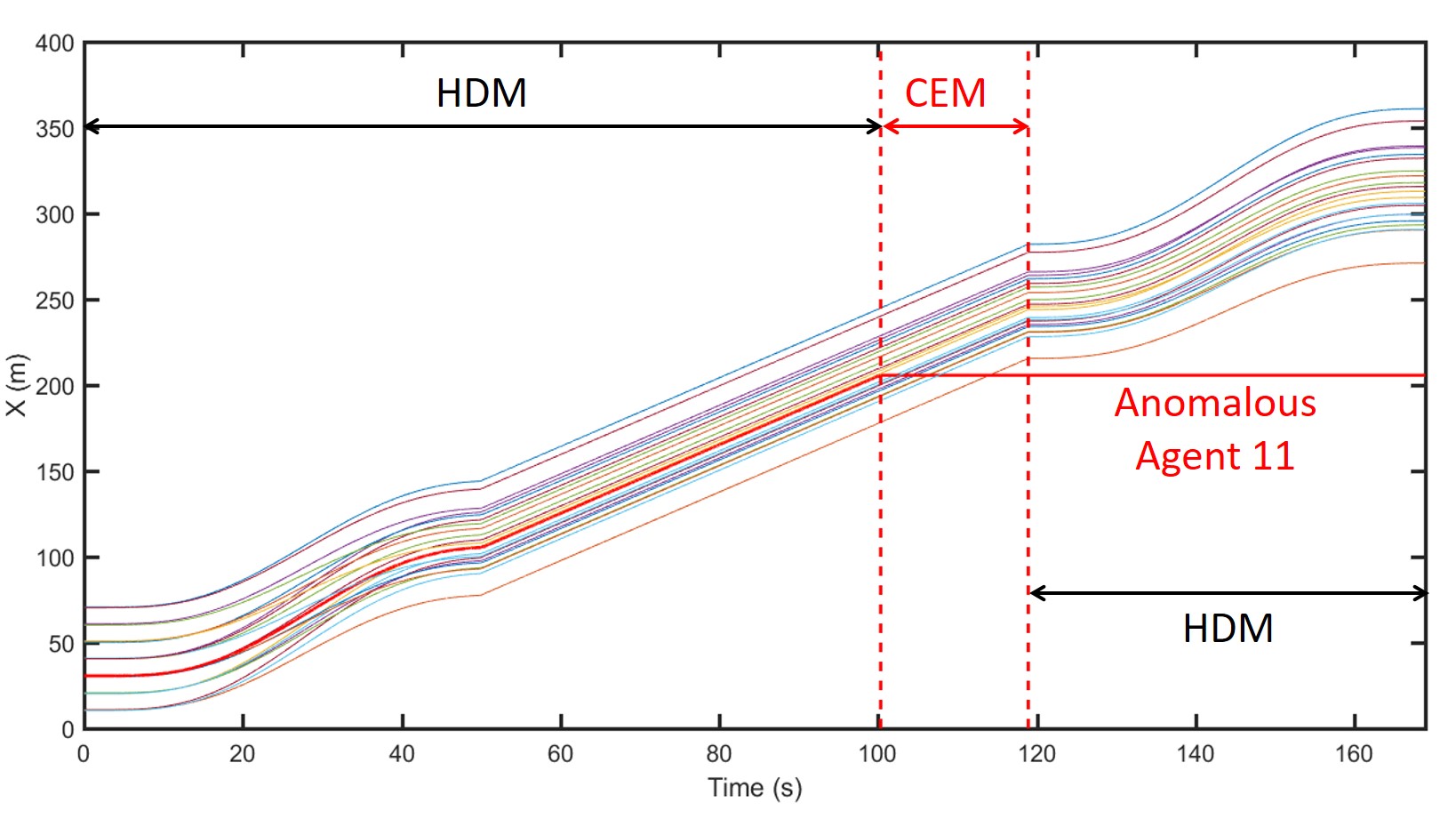}}
\subfigure[]{\includegraphics[width=0.9\linewidth]{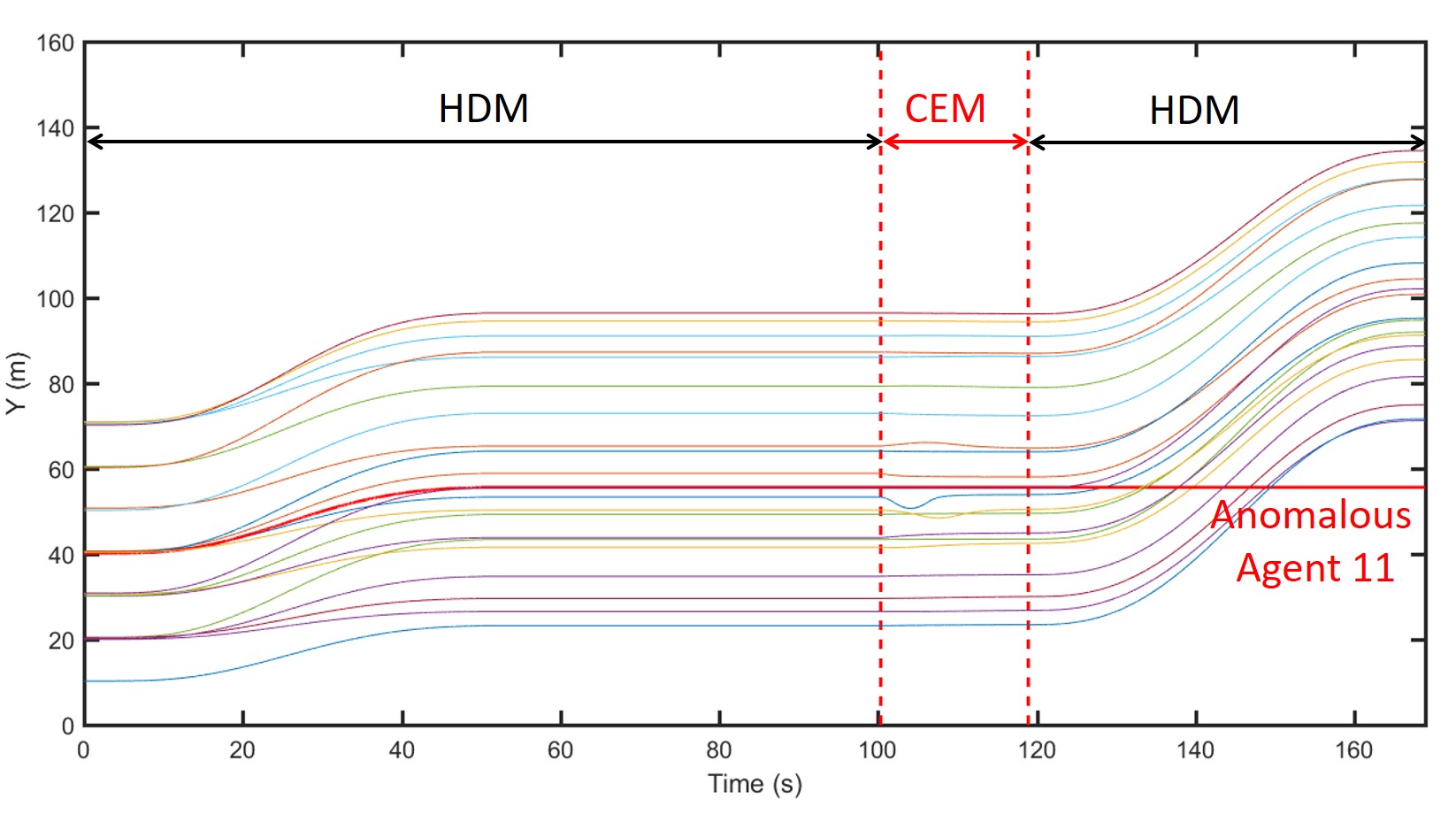}}
\vspace{-.4cm}
\caption{(a,b) $x$ and $y$ components of actual positions of agents versus time for $t\in [0,168.91]s$. HDM is initially active over $t\in [0,100]s$. Agent $11$ is flagged anomalous at time $t\in [100,100.35]s$ thus CEM is activated. At $t=118.92s$, agent $11$ is no longer inside the containment box $\Omega_{\mathrm{con}}$. Therefore, HDM is activated.
}
 \label{XYHDMCPMHDMV}
\end{figure}

\begin{figure}[ht]
\centering
\includegraphics[width=3.3 in]{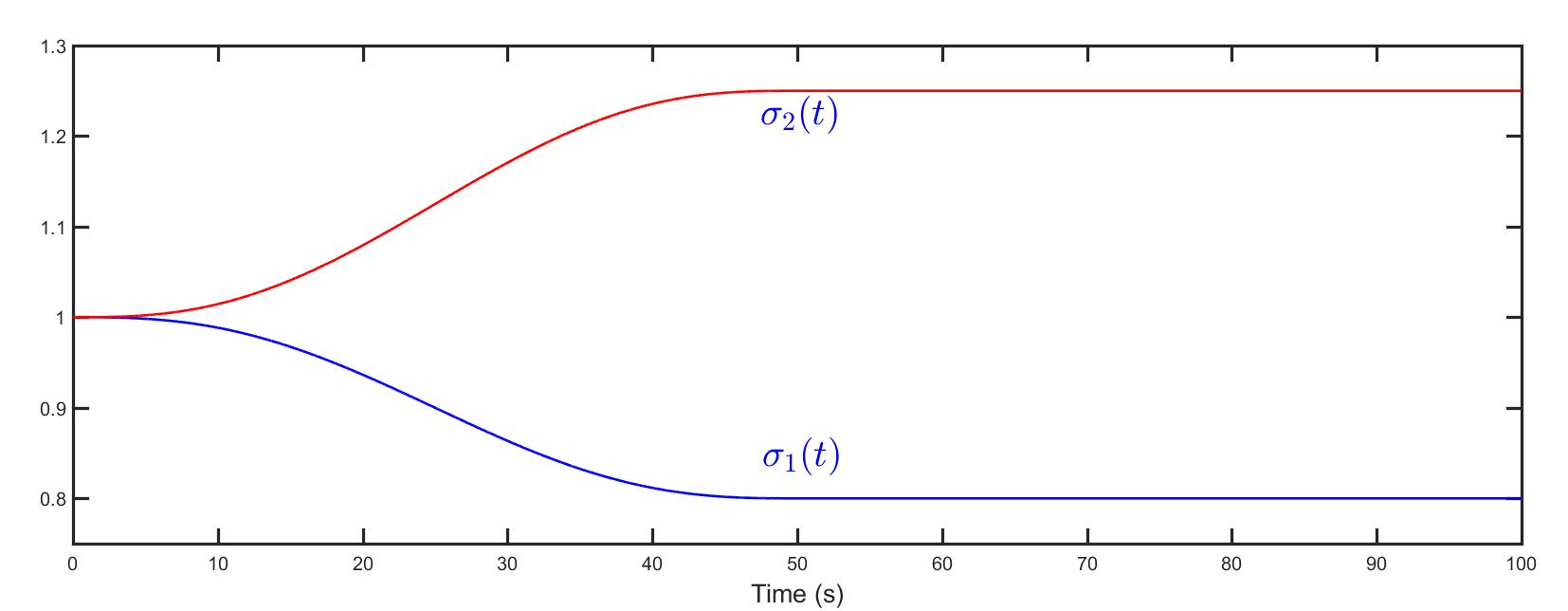}
\caption{Homogeneous deformation eigenvalues $\sigma_1$ and $\sigma_2$ versus time for $t\in [0,100]s$.}
\label{HDMEigenvalues}
\end{figure}
\subsection{Motion Phase 1 (HDM)}
Team collective motion is defined by a homogeneous transformation over $t\in [0,100]$, where agents are all healthy. Agents $i_1=1$, $i_2=2$, and $i_3=3$ are the leaders defining the homogeneous transformation. Given leaders' desired trajectories, eigenvalues of the desired homogeneous deformation coordination, denoted by $\sigma_1$ and $\sigma_2$, are plotted versus time in Fig. \ref{HDMEigenvalues}. Note that $\sigma_3(t)=1$ at any time $t$ because agents are treated as particles of a $2$-D continuum and the desired homogenous deformation coordination is also two dimensional. Follower vehicles apply the communication graph shown in Fig. \ref{InitialFormation} to acquire the desired coordination by local communication. The communication graph is strictly $3$-reachable per Section \ref{Homogeneous Deformation Mode}. Given initial positions of all agents, every follower chooses three in-neighbor agents using the approach described in Section \ref{Homogeneous Deformation Mode}. Consequently, the graph shown in Fig. \ref{Homogeneous Deformation Mode} assigns inter-agent communication, where  followers' communication weights are consistent with agents' positions at reference time $t=0$ and obtained by \eqref{ComWeights}. As shown Fig. \ref{PROOFMAIN}, HDM is active before an anomaly situation arises at time $t=100s$. 
\subsection{Motion Phase 2 (CEM)}
Suppose agent $11$ fails at time $t=100$. This failure is quickly detected by the team using the distributed failure detection method developed in Section  \ref{Continuum Deformation Anomaly Management}. As shown in Figs.  \ref{Weightsss} (a),(c),(d), conditions $\varpi_{11,13}^{\mathrm{min}}(t)\leq  w_{11,13}\leq \varpi_{11,13}^{\mathrm{max}}(t)$, $\varpi_{11,8}^{\mathrm{min}}(t)\leq  w_{11,8}\leq \varpi_{11,8}^{\mathrm{max}}(t)$, and $\varpi_{11,6}^{\mathrm{min}}(t)\leq  w_{11,6}\leq \varpi_{11,6}^{\mathrm{max}}(t)$ $\varpi_{11}$ are satisfied over $t\in [0,100]s$. However, condition $\varpi_{11,6}^{\mathrm{min}}(t)\leq  w_{11,6}\leq \varpi_{11,6}^{\mathrm{max}}(t)$ is violated at $t=100.34$ when $\varpi_{11,6}^{\mathrm{min}}(100.34)>w_{11,6}$. Therefore, CEM is activated, and healthy agent coordination is treated as an ideal fluid flow after $100.35s$. The ideal fluid flow coordination is defined by combining ``Uniform'' and ``Doublet'' flow patterns. Anomalous agent $11$ is wrapped by a disk of radius $a=4m$ resulted from choosing $u_\infty=10$, and {\color{black}$\delta=160$}, i.e. {\color{black}$a=\sqrt{\dfrac{\delta}{u_\infty}}=4m$}.  The remaining healthy vehicles slide along level curves $\psi_{i,c}(t)=\psi_{i,0}$, where each $\psi_{i,0}$ is determined based on agent $i$'s position at $t=100.35s$.

\begin{figure}[ht]
\centering
\subfigure[]{\includegraphics[width=0.73\linewidth]{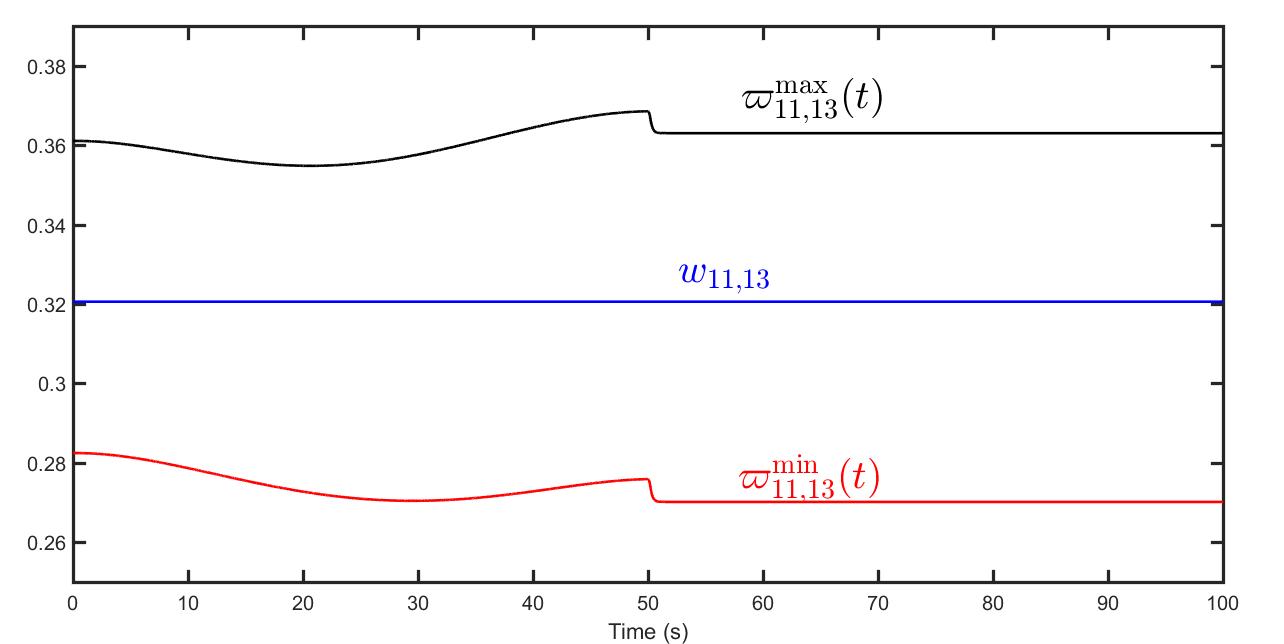}}
\subfigure[]{\includegraphics[width=0.25\linewidth]{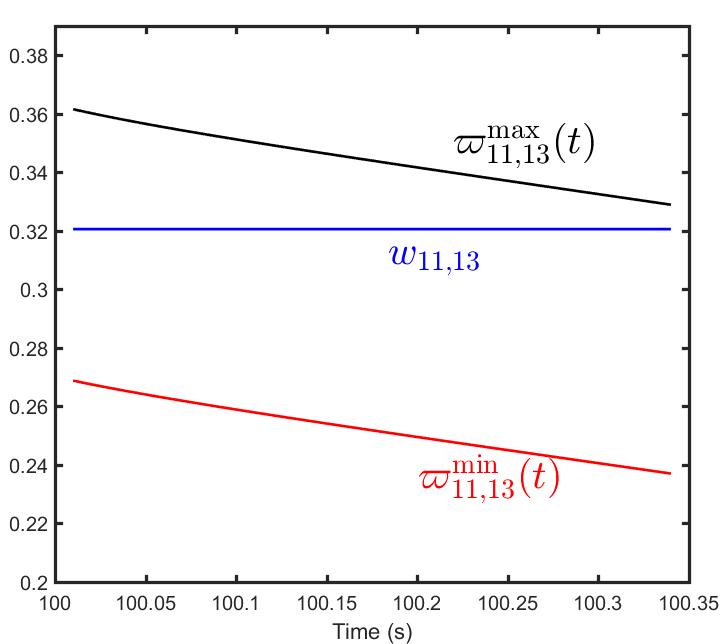}}
\subfigure[]{\includegraphics[width=0.73\linewidth]{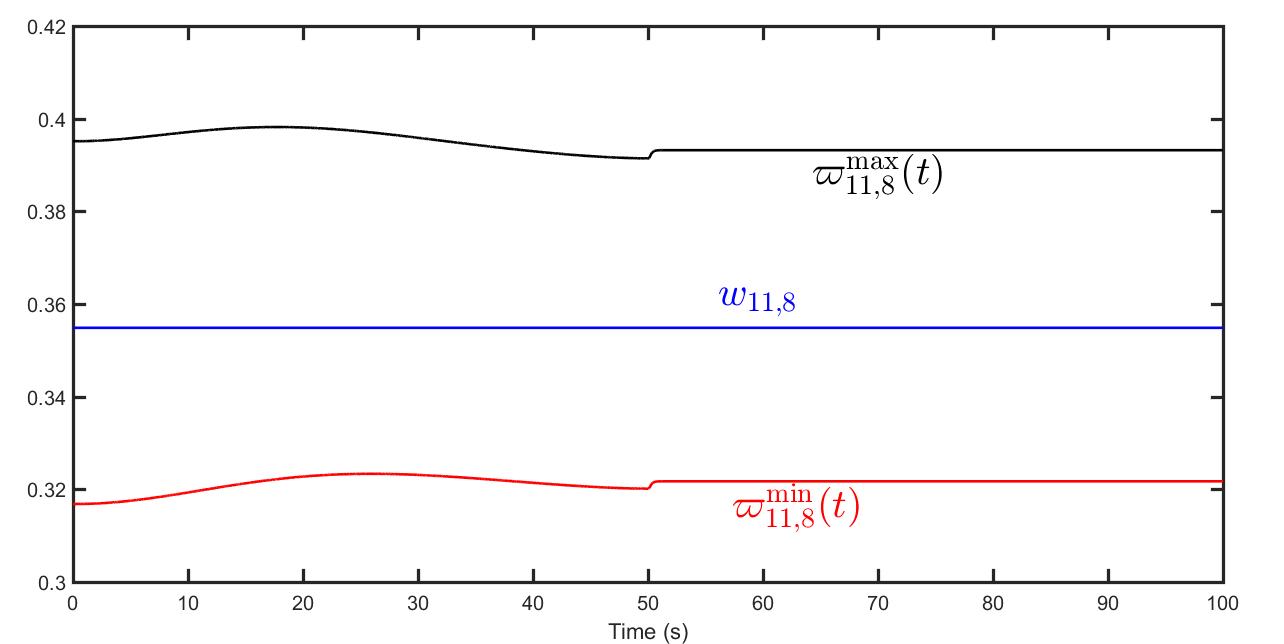}}
\subfigure[]{\includegraphics[width=0.25\linewidth]{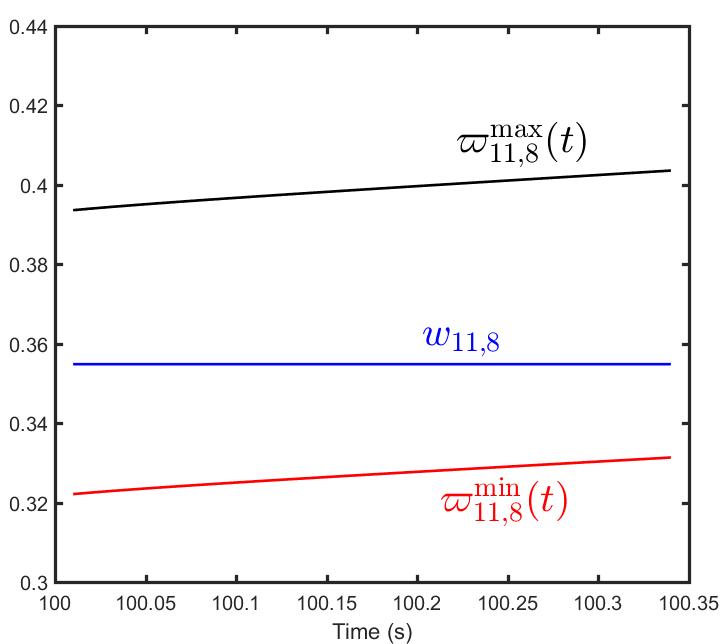}}
\subfigure[]{\includegraphics[width=0.73\linewidth]{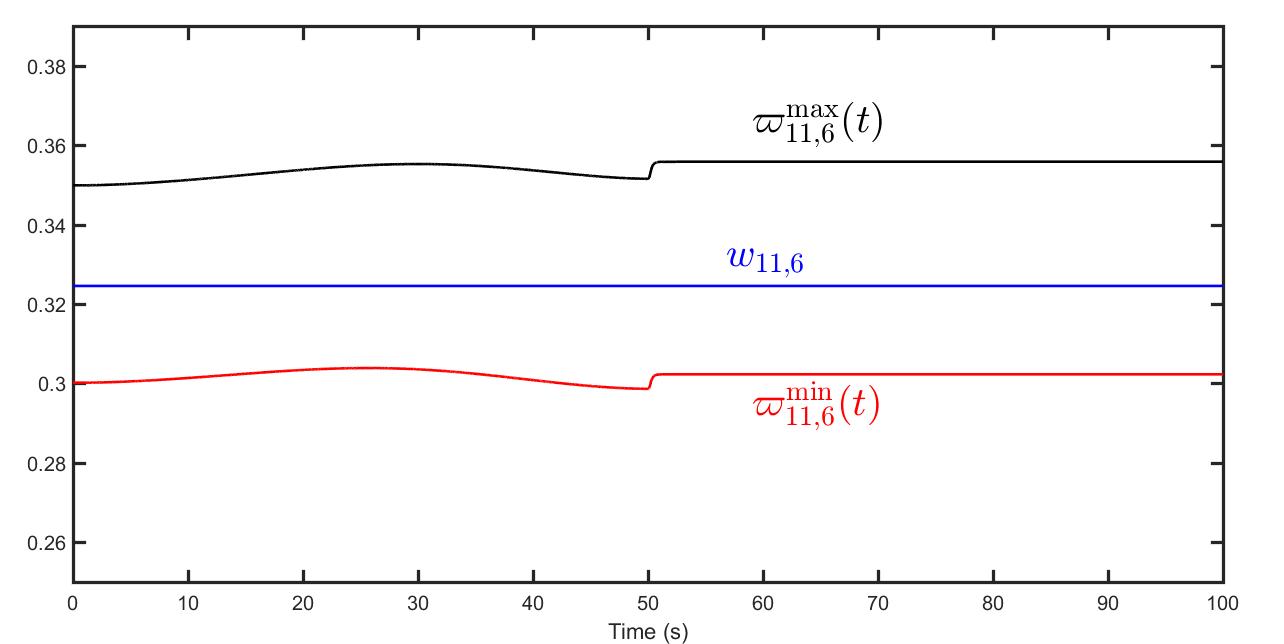}}
\subfigure[]{\includegraphics[width=0.25\linewidth]{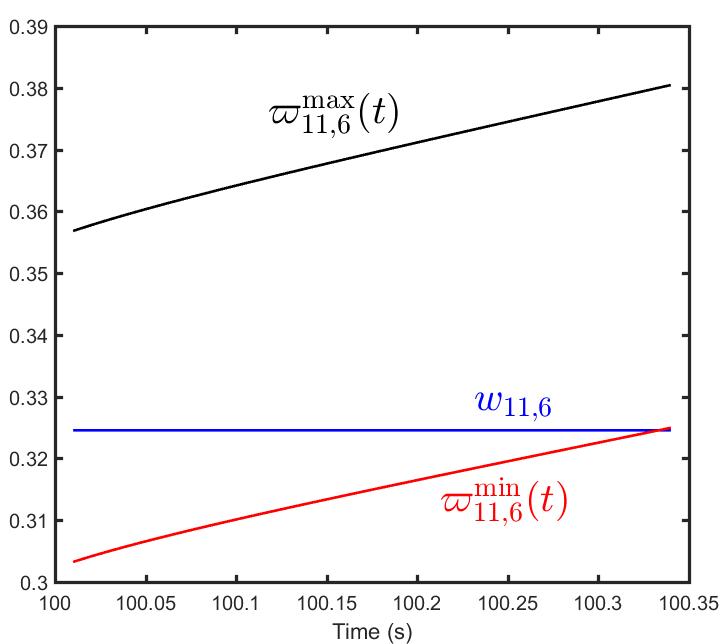}}
\vspace{-.4cm}
\caption{Weights $w_{11,13}$, $\varpi_{11,13}^{\mathrm{min}}(t)$, and $\varpi_{11,13}^{\mathrm{max}}(t)$ for (a) $t\in [0,100]$ and (b) $t\in [100.01,100.35]$. Weights $w_{11,8}$, $\varpi_{11,8}^{\mathrm{min}}(t)$, and $\varpi_{11,8}^{\mathrm{max}}(t)$ for (c) $t\in [0,100]$ and (d) $t\in [100.01,100.35]$. Weights $w_{11,6}$, $\varpi_{11,6}^{\mathrm{min}}(t)$, and $\varpi_{11,6}^{\mathrm{max}}(t)$ for (e) $t\in [0,100]$ and (f) $t\in [100.01,100.35]$. Anomalous motion in agent $11$ is detected in $0.34s$ when $\varpi_{11,6}^{\mathrm{min}}(100.34)>w_{11,6}$.}
 \label{Weightsss}
\end{figure}
In Fig. \ref{CEMEvolution}, actual paths of the healthy agents, defined by $\mathcal{V}_H=\{1,\cdots,10,12,\cdots,22\}$ are shown for $t\in  [100.35,118.92]$. Green markers show positions of healthy agents at $t=100.35s$ when they enter CEM;  black markers show positions of healthy agents at $t=118.92s$ when CEM ends. Failed agent $11$ is wrapped by a disk of radius {\color{black}$4m$} centered at $(205.26,55.62)$ in this example.
\begin{figure}[ht]
\centering
\includegraphics[width=3.3 in]{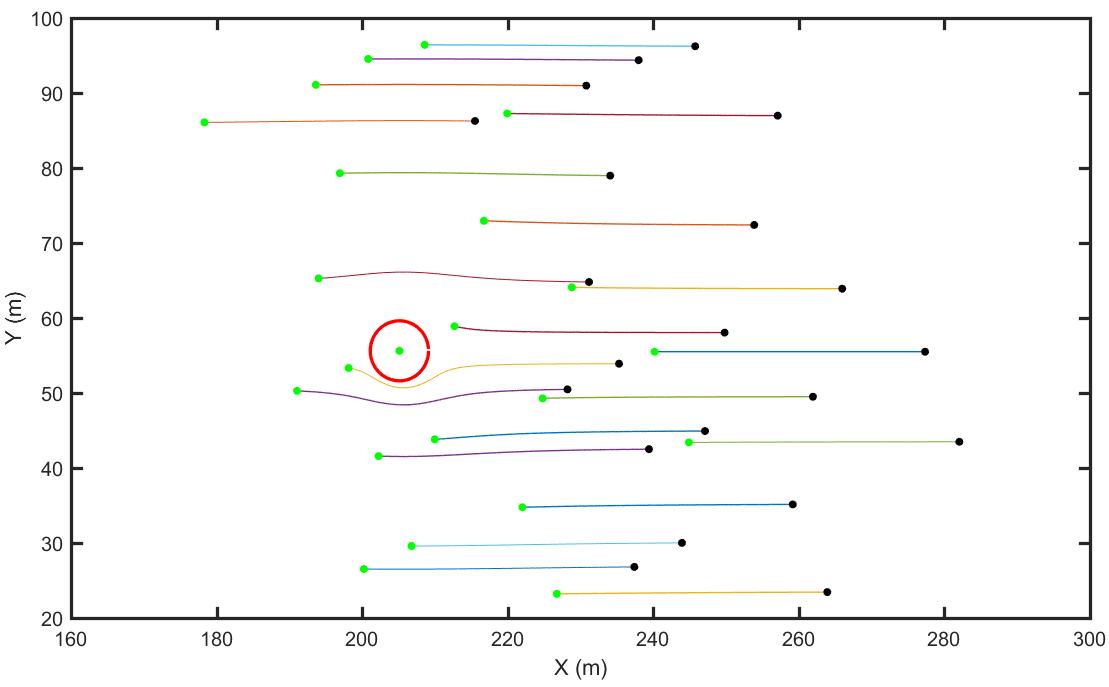}
\caption{Paths of healthy agents over $t\in [100.35,118.92]$ when CEM is active. The green and block markers show positions of agents at times $100.35$ and $118.92$, respectively. Failed agent $11$ is wrapped by a disk of radius {\color{black}$4m$} centered at $(205.26,55.62)$ when CEM is active.}
\label{CEMEvolution}
\end{figure}

\subsection{Motion Phase 3 (HDM)}
CEM continues until switching time $118.92s$ when failed agent $11$ leaves containment box $\Omega_{\mathrm{con}}$. Fig. \ref{FinalFormation} shows the agents' configuration at time $t=118.92$. Followers use the method from Section \ref{Homogeneous Deformation Mode} to find their in-neighbors as well as communication weights. HDM remains active after $t=118.92$ since no other agents fail in this simulation. $x$ and $y$ components of actual agent positions were plotted versus time for $t\in[118.92,168.92]s$ earlier in Fig.  \ref{XYHDMCPMHDMV}.
\begin{figure}[ht]
\centering
\includegraphics[width=3.3 in]{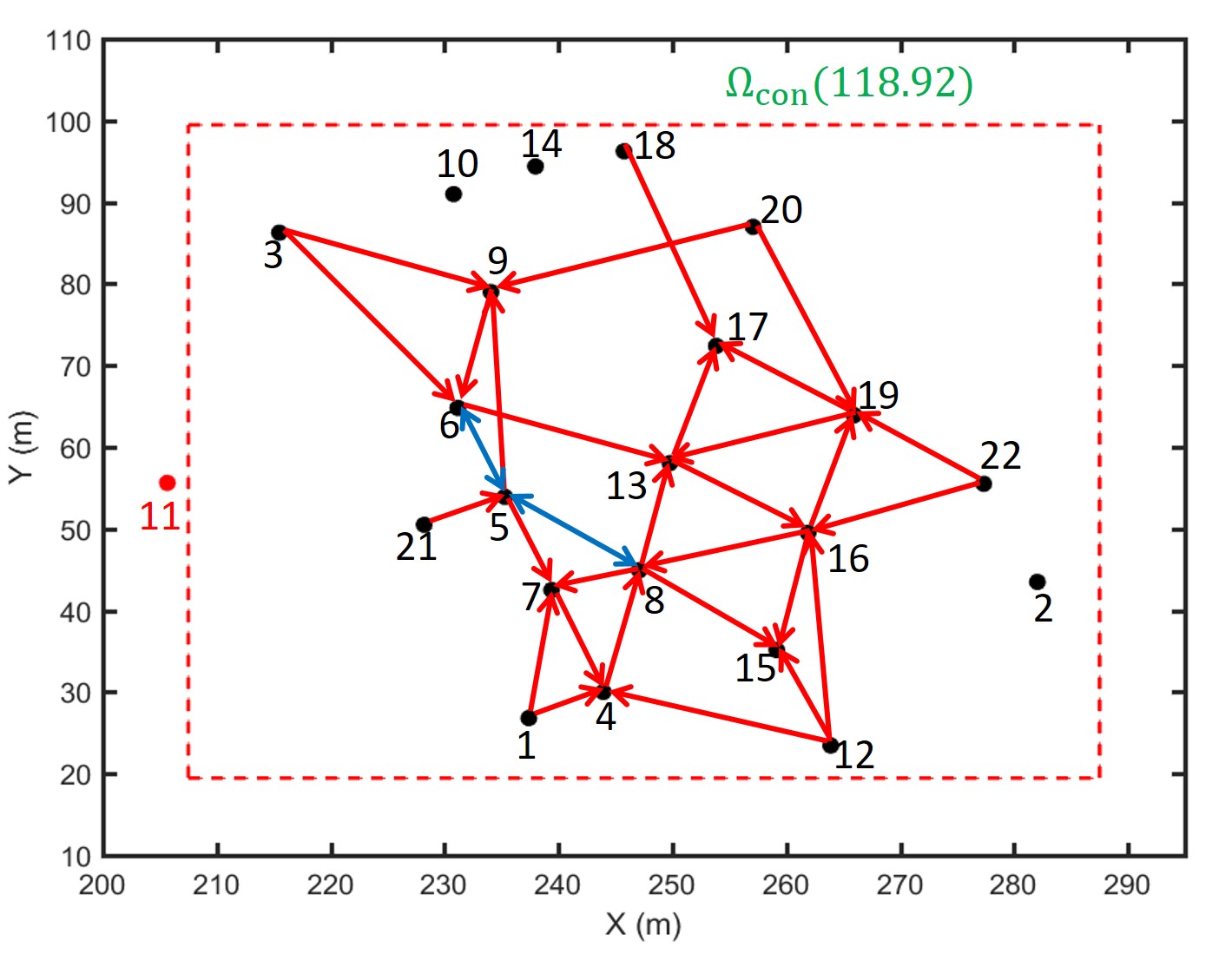}
\caption{Formation of agents at time $t=118.92s$. Failed agent $11$ is outside the containment region $\Omega_{\mathrm{con}}$. HDM is activated and inter-agent communication is established using the procedure developed in Section \ref{Homogeneous Deformation Mode}. }
\label{FinalFormation}
\end{figure}

\section{Conclusion}\label{conclusion}
This paper develops a hybrid cooperative control strategy with two operational modes to manage large-scale coordination of agents in a resilient fashion. The first mode (HDM) treats agents as particles of a deformable body and is active when all agents are healthy. HDM guarantees agents can safely initialize and coordinate their motions using the unique features of homogeneous deformation coordination. A new CEM cooperative paradigm was proposed to handle cases in which one or more vehicles in the shared motion space fail to admit the desired coordination.   In CEM the desired vehicle coordination is treated as an ideal fluid flow and failed vehicles are excluded by closed curves. Therefore, desired trajectories for the remaining healthy vehicles can be planned and collective motion safety for healthy vehicles can still be guaranteed with low computation overhead. To automatically initiate transition to CEM, this paper contributes a strategy for quickly detecting agent failure using the unique properties of the homogeneous deformation coordination.  Future work is needed to relax motion constraints on failed vehicles and present simulation results with realistic vehicle dynamics and more complex environments.

\begin{ack}                               
This work has been supported by the National Science Foundation under Award Nos. 1739525 and 1914581.
\end{ack}

\bibliographystyle{plain}        
\bibliography{myrefs}           



\end{document}